\documentclass[a4paper,UKenglish,final]{lipics-v2021}

\hideLIPIcs

\title{Strategy Synthesis for Global Window PCTL} %

\author{Benjamin Bordais}{Université Paris-Saclay, CNRS, ENS Paris-Saclay, LMF, 91190 Gif-sur-Yvette, France}{bordais@lsv.fr}{}{}
\author{Damien {Busatto-Gaston}}{Universit\'e libre de Bruxelles, Brussels, Belgium \and \url{https://di.ulb.ac.be/verif/dbusatto/} }{damien.busatto-gaston@ulb.be}{https://orcid.org/0000-0002-7266-0927}{} %
\author{Shibashis Guha}{Tata Institute of Fundamental Research, Mumbai, India \and
\url{https://www.tifr.res.in/~shibashis.guha/} }{shibashis.guha@tifr.res.in}{https://orcid.org/0000-0002-9814-6651}{}
\author{Jean-Fran\c{c}ois Raskin}{Universit\'e libre de Bruxelles, Brussels, Belgium}{jean-francois.raskin@ulb.be}{}{}

\authorrunning{B. Bordais, D. Busatto-Gaston, S. Guha, J.-F. Raskin} %

\Copyright{Benjamin Bordais, Damien Busatto-Gaston, Shibashis Guha and Jean-Fran\c{c}ois Raskin} %



\funding{This work is partially supported by the ARC project Non-Zero Sum Game Graphs: Applications to Reactive Synthesis and Beyond (Fédération Wallonie-Bruxelles), the EOS project Verifying Learning Artificial Intelligence Systems (F.R.S.-FNRS \& FWO), the COST Action 16228 GAMENET (European Cooperation in Science and Technology), by the PDR project Subgame perfection in graph games (F.R.S.- FNRS), and the DST-SERB SRG/2021/000466 "Zero-sum and Nonzero-sum Games for Controller Synthesis of Reactive Systems" project.}

\nolinenumbers %

\EventEditors{John Q. Open and Joan R. Access}
\EventNoEds{2}
\EventLongTitle{42nd Conference on Very Important Topics (CVIT 2016)}
\EventShortTitle{CVIT 2016}
\EventAcronym{CVIT}
\EventYear{2016}
\EventDate{December 24--27, 2016}
\EventLocation{Little Whinging, United Kingdom}
\EventLogo{}
\SeriesVolume{42}
\ArticleNo{23}

\usepackage{amsfonts}

\usepackage{hyperref}

\usepackage{xspace} %
\usepackage{xparse} %
\usepackage{csquotes}

\usepackage{stmaryrd} %

\usepackage{tikz}
\usetikzlibrary{arrows,shapes,automata,positioning}
\tikzset{>=latex,initial text=}

\tikzstyle{player0}=[state,draw,rounded rectangle,align=center,minimum size=6mm]
\tikzstyle{player1}=[state,draw,rectangle,align=center] %
\tikzstyle{player2}=[state,draw,rounded rectangle,align=center] %

\tikzstyle{action0}=[->]
\tikzstyle{action1}=[->,color=blue]
\tikzstyle{action2}=[->,color=red!80!black]
\tikzstyle{action3}=[->,color=green!80!black]
\tikzstyle{action4}=[->,color=orange!80!black]
\tikzstyle{action5}=[->,color=purple!80!black]

\tikzset{every loop/.style={looseness=7}} %

\usepackage[basic]{complexity}
\renewclass\P{PTIME}
\renewclass\EXP{EXPTIME}
\renewclass\NEXP{NEXPTIME}
\newclass\APSPACE{APSPACE}

\newcommand\newmath[2]{\newcommand#1{\ensuremath{#2}\xspace}}
\newcommand\renewmath[2]{\renewcommand#1{\ensuremath{#2}\xspace}}
\newcommand\newmathope[2]{\newcommand#1{\ensuremath{\operatorname{#2}}\xspace}}

\newcommand\newmathopelimits[2]{\newcommand#1{\ensuremath{\operatornamewithlimits{#2}}\xspace}}

\newmath{\N}{\mathbb{N}}
\newmath{\Z}{\mathbb{Z}}
\newmath{\Q}{\mathbb{Q}}
\renewmath{\R}{\mathbb{R}}

\newmath{\Nspos}{\mathbb{N}_{>0}}
\newmath{\Qspos}{\mathbb{Q}_{>0}}
\newmath{\Rspos}{\mathbb{R}_{>0}}
\newmath{\Qpos}{\mathbb{Q}_{\geq 0}}
\newmath{\Rpos}{\mathbb{R}_{\geq 0}}

\newmathopelimits{\argmin}{\arg\min}
\newmathopelimits{\argmax}{\arg\max}

\renewmath{\Pr[1]}{\mathbb P\left[{#1}\right]}
\newmath{\PrBig[1]}{\mathbb P\Bigl[{#1}\Bigr]}
\newmath{\PrSmall[1]}{\mathbb P[{#1}]}
\newmath{\Dist}{\mathsf{Dist}}
\newmath{\Supp}{\mathsf{Supp}}
\renewmath{\E[1]}{\mathbb E\left[{#1}\right]}

\newmath{\zug[1]}{\left\langle{#1}\right\rangle}

\newmath{\sem[1]}{\left\llbracket{#1}\right\rrbracket}

\newmath{\MC}{M}
\newmath{\MDP}{\mathcal M}

\newmath{\FPaths}{\mathsf{FPaths}}
\newmath{\IPaths}{\mathsf{Paths}}
\newmath{\first}{\mathsf{first}}
\newmath{\last}{\mathsf{last}}

\newmath{\Cyl}{\mathcal{C}}

\renewmath{\AP}{\mathsf{AP}}
\newmath{\f}{\mathsf{f}}
\newmath{\proba}{\mathbb{P}}
\newmath{\initState}{s_{{\sf init}}}
\newmath{\MDPtoMC[2]}{{#1}[{#2}]}

\newmath{\MDPstrat}{\MDPtoMC{\MDP}{\sigma}}
\newmath{\probastrat}{\mathbb{P}_{\sigma}}

\newmath{\ellmax}{\ell_{\max}}

\newmath{\strat}{\sigma}
\newmath{\stratW}{\partial}

\newmath{\MD}{\mathsf{MD}}
\newmath{\HD}{\mathsf{HD}}
\newmath{\MR}{\mathsf{MR}}
\newmath{\HR}{\mathsf{HR}}

\newmathope{\XX}{X}
\newmathope{\UU}{U}
\newmathope{\WW}{W}
\newmathope{\GG}{G}
\newmathope{\FF}{F}
\newmathope{\Forall}{A}
\newmathope{\Exists}{E}

\newmath{\Ell}{\mathcal L}
\newmath{\FO}{\text{FO}}
\newmath{\POLY}{\text{POLY}}
\newmath{\STATER}{\text{STATE0}}
\newmath{\PATH}{\text{PATH}}
\newmath{\STATE}{\text{STATE1}}

\newmath{\wmp}{\mathsf{WMP}}
\newmath{\sizeMin}{\ell}

\allowdisplaybreaks

\begin{document}
\maketitle

\begin{abstract}
  Given a Markov decision process (MDP) $M$ and a formula $\Phi$,
  the strategy synthesis problem asks if there exists a strategy $\sigma$ s.t.
  the resulting Markov chain $M[\sigma]$ satisfies $\Phi$.
  This problem is
  known to be undecidable for the probabilistic temporal logic PCTL.
  We study a class of formulae that can be seen as a
  fragment of PCTL where a local, bounded horizon property
  is enforced all along an execution.
  Moreover, we allow for linear expressions in the probabilistic inequalities.
  This logic is at the frontier of
  decidability, depending on the type of strategies considered.
  In particular, strategy synthesis is decidable when strategies
  are deterministic while the general problem is undecidable.
\end{abstract}

\section{Introduction}

  Given an MDP $M$ and a probabilistic temporal logic formula $\Phi$,
  the strategy synthesis problem is to determine if there exists
  a strategy $\sigma$ to resolve the nondeterminism in $M$ such that
  the resulting Markov chain (MC) $M[\sigma]$ satisfies $\Phi$, and if so,
  to construct one such strategy.
  The probabilistic temporal logic that we study in this paper allows us to express
  rich probabilistic global temporal constraints over a
  \emph{bounded} horizon that must be enforced along all computations.
  Let us illustrate our logic with a few examples.
  The formula
  $\Forall \GG ( \proba(\FF^{5} {\sf Good}) \geq 0.95)$
  expresses that it must always be the case, under the strategy $\sigma$,
  that along all computations, the probability to reach a good state
  within $5$ steps is at least $0.95$. This is a quantitative bounded horizon
  B\"uchi property. In addition, our logic allows for comparing the probability
  of different events:
  $\Forall \GG \proba(\FF^{5} {\sf Good}) \geq 2 \times  \proba(\FF^{10} {\sf Bad}))$
  expresses that under the strategy $\sigma$, along all computations,
  it is always the case that the probability to reach a good state within $5$ steps
  is at least twice the probability of reaching a bad state within $10$ steps.
  The ability to compare probabilities of different events, while not present
  in classical logics like PCTL, is necessary to express properties
  like probabilistic noninterference~\cite{Gray92}.
  This feature has been introduced and studied in probabilistic hyperlogics, e.g.~\cite{AB18},  where the ability to compare probabilities plays a central role in describing applications. While hyperlogics are very expressive and highly undecidable, we study here the ability to compare probabilities in a weaker logical setting in order to understand more finely the decidability border that probabilistic comparisons, and more generally linear expressions, induce.

  While the model-checking problem for PCTL
  and MDPs
  is decidable~\cite{BK08}, the synthesis problem is in general
  undecidable~\cite{10.1109/LICS.2006.48}.\footnote{%
  The difference between the two problems is essentially as follows:
  in the model-checking problem, each probabilistic operator in the formula
  is associated with one strategy (or scheduler) while in the synthesis problem,
  a \emph{unique} strategy is fixed and used for all the probabilistic operators.}
  Synthesis for PCTL~\cite{10.1109/LICS.2006.48}, HyperPCTL~\cite{AB18}, and
  as well as for our logic (as shown in Theorem~\ref{thm:gwpctl_co_re_hard})
  is undecidable. To recover decidability, we explore two options.
  First, we consider subclasses of strategies:
  memoryless deterministic strategies (\MD), memoryless randomized strategies (\MR),
  and history-dependant deterministic strategies (\HD) are important classes to be
  considered.
  Second, we identify syntactically defined sublogics
  with better decidability properties.
  For instance, while for PCTL objectives the synthesis problem for
  \HD strategies is highly undecidable
  ($\Sigma_1^1$-complete)~\cite{10.1109/LICS.2006.48}, it has been shown
  that the problem is decidable for the cases
  of \MD and \MR strategies~\cite{10.1109/LICS.2006.48}. The synthesis problem
  for the qualitative fragment of PCTL, where probabilistic operator
  can only be compared to constant $0$ and $1$, is decidable for \HD strategies.
  An important contribution of this paper is to show that the synthesis problem for
  our sublogics is decidable for \HD strategies. To the best of our knowledge,
  this is the first decidability result for a class of unbounded memory strategies
  (here \HD) and quantitative probabilistic temporal properties.

\subparagraph*{{\bf Main technical contributions}}
  We introduce the logic $L$-PCTL and two sublogics. $L$-PCTL extends PCTL
  with linear constraints over probability subformulae.
  We first study the window $L$-PCTL fragment that only allows bounded until
  or bounded weak until operators in the path formulae. The results for this
  fragment are presented in Table~\ref{tbl:global-window}(\subref{tbl:W-LPCTL})
  where columns distinguish between memoryless ({\sf M}) and history-dependent
  strategies ({\sf H}), and rows between deterministic ({\sf D})
  and randomized strategies ({\sf R}).
  Second, we study the \emph{global} window $L$-PCTL extension of this logic in which
  window formulae appear in the scope of an $\Forall \GG$ operator that imposes
  the window formula to hold on every state of every computation.
  The results for this fragment are presented
  in Table~\ref{tbl:global-window}(\subref{tbl:GW-LPCTL}).
  Third, we adapt results from the literature to the full logic as summarized
  in Table~\ref{tbl:global-window}(\subref{tbl:LPCTL}).
  An $L$-PCTL formula is \emph{flat} if it does not have nested probabilistic
  operators, while it is \emph{non-strict} if it does not contain strict
  comparison operators ($>$ or $<$)
  for comparing probability expressions. %

\begin{table}[t]

\caption{A summary of our results for the synthesis problem
  on MDPs for $L$-PCTL formulae.}
\centering
\begin{subfigure}[t]{0.5\textwidth}
\caption{synthesis for window $L$-PCTL}
\label{tbl:W-LPCTL}
\centering
\begin{tabular}{r|c|c}
   & {\sf M} & {\sf H} \\ \hline
  \begin{tabular}[c]{@{}c@{}}~\\{\sf D}\\~\end{tabular} &
    \begin{tabular}[c]{@{}c@{}}$\NP$-complete\\\cite{10.1007/1-4020-8141-3_38}\end{tabular} &
    \begin{tabular}[c]{@{}c@{}} $\PSPACE$-complete\\Prop.~\ref{prop:HD-window}\end{tabular} \\ \hline
  {\sf R} &
    \begin{tabular}[c]{@{}c@{}} $\PSPACE$\\\textsc{Sqrt-Sum}-hard\\Prop.~\ref{prop:memoryless-window-synthesis}, \cite{10.1109/LICS.2006.48}\end{tabular} &
    \begin{tabular}[c]{@{}c@{}} $\EXPSPACE$\\$\PSPACE$-hard\\Thm.~\ref{thm:window-synthesis}\end{tabular}
\end{tabular}
\end{subfigure}\hfill%
\begin{subfigure}[t]{0.5\textwidth}
\caption{synthesis for global window $L$-PCTL}
\label{tbl:GW-LPCTL}
\centering
\begin{tabular}{r|c|c}
   & {\sf M} & {\sf H} \\ \hline
  {\sf D} &
    $\NP$-complete &
    \begin{tabular}[c]{@{}c@{}} $2\EXP$\\$\EXP$-hard\\Prop.~\ref{prop:HD-global-window}\end{tabular} \\ \hline
  {\sf R} &
    \begin{tabular}[c]{@{}c@{}} $\PSPACE$\\\textsc{Sqrt-Sum}-hard\\Prop.~\ref{prop:MR-global-window}\end{tabular} &
    \begin{tabular}[c]{@{}c@{}} $\co\RE$-complete\footnote{if the formula is flat and non-strict}\\$\Sigma_1^1$-hard\\Thm.~\ref{thm:coRE-easy}, \ref{thm:gwpctl_co_re_hard}\end{tabular}
\end{tabular}
\end{subfigure}\hfill
\begin{subfigure}[t]{0.6\textwidth}%
\caption{synthesis for $L$-PCTL}
\label{tbl:LPCTL}
\centering
\begin{tabular}{r|c|c}
  strategies: & Memoryless & History-dependent \\ \hline
  \begin{tabular}[c]{@{}c@{}}~\\Deterministic\\~\end{tabular} &
    \begin{tabular}[c]{@{}c@{}}$\NP$-complete\\\cite{10.1007/1-4020-8141-3_38}\end{tabular} &
    \begin{tabular}[c]{@{}c@{}}$\Sigma_1^1$-complete\\\cite{10.1109/LICS.2006.48}\end{tabular} \\ \hline
  Randomized &
    \begin{tabular}[c]{@{}c@{}} $\EXP$\\\textsc{Sqrt-Sum}-hard\\\cite{10.1007/11590156_44}, \cite{10.1109/LICS.2006.48}\end{tabular} &
    \begin{tabular}[c]{@{}c@{}}$\Sigma_1^1$-hard\\\cite{10.1109/LICS.2006.48}\end{tabular}
\end{tabular}
\end{subfigure}

\label{tbl:global-window}
\end{table}

  Our two main technical contributions are focused on the synthesis problem
  for the global window $L$-PCTL logic. First, we introduce
  a fixpoint characterization of the set of strategies %
  that enforces an $L$-PCTL window property globally. This
  characterization is effective for \HD strategies, leads to a $2\EXP$ algorithm,
  and we provide an $\EXP$ lower bound. Furthermore, the fixpoint characterization
  allows us to prove that the synthesis problem is in $\co\RE$ for the class of
  history-dependent randomized (\HR) strategies for the flat and non-strict fragment
  of global window $L$-PCTL. Second, we prove that the synthesis problem
  for \HR strategies is undecidable with an original technique that reduces
  the halting problem of 2-counter Minsky machines (2CM) to our synthesis problem.
  We believe that the fixpoint characterization and the 2CM encoding are
  of independent interest.

  Finally, the satisfiability problem~\cite{DBLP:conf/lics/BrazdilFKK08} for PCTL
  (and its variants) can be reduced to the synthesis problem.
  The decidability of the satisfiability problem for PCTL is a long standing
  open problem. Our decidability result for the synthesis problem for \HD strategies
  and global window PCTL formulae can be transferred to the following version of the
  satisfiability problem: given a granularity $g$ for the probabilities, and
  a global window PCTL formula $\Phi$, does there exist an MC with granularity $g$
  that satisfies $\Phi$? (Theorem~\ref{thm:satisfiability}).
  This gives a new positive decidability result for the satisfiability problem with an unbounded horizon fragment of PCTL and unbounded MCs.

\subparagraph*{{\bf Related work}}
  The model-checking problem for PCTL is decidable~\cite{BK08} and should not
  be confused with the synthesis problem. In \cite{10.1109/LICS.2006.48}, the authors
  study the synthesis problem for PCTL on MDPs and stochastic games.
  In~\cite{10.1007/1-4020-8141-3_38} it is shown that both randomization
  and memory
  in strategies are necessary even for flat window PCTL formulae.
  Further, \cite{10.1007/1-4020-8141-3_38} shows that the synthesis of \MD strategies
  for PCTL objectives is $\NP$-complete, and \cite{10.1007/11590156_44} shows that \MR synthesis is in $\EXP$.\footnote{For the existence of \MR strategies
  for PCTL objectives, in the introduction of~\cite{10.1109/LICS.2006.48},
  it is claimed that the problem is in $\PSPACE$, with a reference
  to~\cite{10.1007/11590156_44}.
  However, in \cite{10.1007/11590156_44} only an $\EXP$ upper bound is proven,
  for the more general problem of stochastic games.
  The proof encodes the problem as a polynomial-size formula in the first-order
  theory of the reals with a fixed alternation of quantifiers
  so that deciding it is in $\EXP$.~%
  The claim seems to be that the complexity of their approach drops to $\PSPACE$
  when all states are controllable.
  There is no convincing argument there for that claim, in particular their formula
  still contains universal quantifiers.
  }
  For the qualitative fragment of PCTL, deciding the existence of \MR and \HD
  strategies have been shown to be $\NP$-complete and $\EXP$-complete,
  respectively~\cite{10.1109/LICS.2006.48}.
  As previously mentioned, the synthesis problems for \HD and \HR strategies
  in the general case of (quantitative) PCTL objectives are highly undecidable~\cite{10.1109/LICS.2006.48}.

  In~\cite{DFT20}, a probabilistic hyperlogic (PHL) has been introduced
  to study hyperproperties of MDPs.
  PHL allows quantification over strategies, and includes PCTL$^*$ and temporal
  logics for hyperproperties such as {\sc HyperCTL$^*$}~\cite{CFKMRS14}.
  Hence the model-checking problem in PHL can ask for the existence of a strategy
  for hyperproperties, and
  has been shown to be undecidable~\cite{DFT20}.
  Another related work is~\cite{ABBD20}, where {\sc HyperPCTL}~\cite{AB18} has been
  extended with strategy quantifiers, and studies hyperproperties over MDPs.
  The model-checking problem for this logic is also undecidable.
  In both of these undecidability proofs,
  the constructed formula contains unbounded finally ($\FF$) properties
  that cannot be expressed in the global window fragment of PCTL that we study here.
  In both \cite{DFT20} and \cite{ABBD20}, the model-checking problem is decidable
  when restricted to \MD strategies but is undecidable for \HD strategies.

  The PCTL satisfiability problem is open for decades.
  In~\cite{KR18}, decidability of finite and infinite satisfiability has been considered for several fragments of PCTL using unbounded finally ($\FF$) and unbounded always ($\GG$) operators.
  In~\cite{CK16}, satisfiability for bounded PCTL has been considered where the number of steps or horizon used in the operators is restricted by a bound.
  In \cite{ACJK21}, a problem related to the satisfiability problem called the feasibility problem has been studied. Given a PCTL formula $\varphi$, and a family of Markov chains defined using a set of parameters and with a fixed number of states, the feasibility problem is to identify a valuation for the parameters such that the realized Markov chain satisfies $\varphi$. In the satisfiability problem that we study here, the number of states is however not fixed a priori and can be arbitrarily large.

\section{Preliminaries}\label{sec:prelims}
  A \emph{probability distribution} on a finite set $S$ is
  a function $d:S\to [0,1]$ such that $\sum_{s\in S}d(s)=1$.
  We denote the set of all probability distributions on set $S$ by $\Dist(S)$.

\begin{definition}
  A \emph{Markov chain} (MC) is a tuple
  $\MC = \zug{S,\initState, \proba, \AP, L}$
  where $S$ is a countable set of states,
  $\initState \in S$ is an initial state,
  $\proba:S\to\Dist(S)$ is a transition function,
  $\AP$ is a non-empty finite set of atomic propositions,
  and $L:S \rightarrow 2^{\AP}$ is a labelling function.
\end{definition}
  If $\proba$ maps a state $s$ to a distribution $d$ so that $d(s')>0$,
  we write $s\xrightarrow{d(s')}s'$ or simply $s\rightarrow s'$,
  and we denote $\proba(s,s')$ the probability $d(s')$.
  We say that the atomic proposition $p$ holds on a state $s$ if $p\in L(s)$.

  A \emph{finite path} $\rho = s_0 s_1 \cdots s_i$ in an MC $\MC$
  is a sequence of consecutive states, so that for all $j\in[0,i-1]$,
  $s_j \rightarrow s_{j+1}$. %
  We denote $|\rho|=i$ the length of $\rho$, $\last(\rho) = s_i$
  and $\first(\rho) = s_0$.
  We also consider states to be paths of length $0$.
  Similarly, an \emph{infinite path} is an infinite sequence
  $\rho = s_0 s_1 \cdots$ so that for all $j\in\N$,
  $s_j \rightarrow s_{j+1}$. %
  If $\rho$ is a finite (resp. infinite) path $s_0 s_1 \cdots$,
  we let $\rho[i]$ denote $s_i$, $\rho[{:}i]$ denote the finite prefix
  $s_0 \cdots  s_{i}$, and $\rho[i{:}]$ denote the finite (resp. infinite)
  suffix $s_i s_{i+1} \cdots$.

  We denote the set of all finite paths in $\MC$ by $\FPaths_{\MC}$.
  We introduce notations for the subsets $\FPaths^i_{\MC}$
  (resp. $\FPaths^{\leq i}_{\MC}$, $\FPaths^{< i}_{\MC}$)
  of paths of length $i$ (resp. of length at most or less than $i$).
  Let $\FPaths_{\MC}(s)$ denote the set of paths $\rho$ in $\FPaths_{\MC}$
  such that $\first(\rho)=s$.
  More generally, $\FPaths_{\MC}(\rho)$ denotes the set of paths
  which admit $\rho$ as a prefix.
  Similarly, we let $\IPaths_{\MC}$ be the set of infinite paths of $\MC$,
  and extend the previous notations
  for fixing an initial state or a shared prefix.
  In particular, $\IPaths_{\MC}(\rho)$ is called the cylinder of $\rho$.

  If $\rho=s_0\dots s_i$ is a finite path and $\rho'=s_is_{i+1}\dots$ is a finite or infinite path so that
  $\first(\rho')=\last(\rho)$, let $\rho\cdot \rho'=s_0\dots s_is_{i+1}\dots$ denote their concatenation.

\begin{definition}
  Let $s$ be a state of an MC $\MC$.
  The MC $\MC$ naturally defines
  a \emph{probability measure} $\mu_{\MC}^s$ on
  $(\IPaths_{\MC}(s),\Omega_{\MC}^s)$,
  where $\Omega_{\MC}^s$ is the $\sigma$-algebra of cylinders,
  \textit{i.e.}~the sets $\IPaths_{\MC}(\rho)$ with
  $\rho\in\FPaths_{\MC}(s)$, their complements and countable unions.
\end{definition}
  The measure of a cylinder $\IPaths_{\MC}(\rho)$ is the product
  of the probabilities of each transition in the finite path $\rho$,
  and by Carathéodory's extension theorem we get a measure $\mu_{\MC}^s$
  over $\Omega_{\MC}^s$. As $s$ is always obvious from context
  (first state of the paths being considered),
  we omit it from the measure notation, in favour of $\mu_{\MC}$.
  We note that $\FPaths_{\MC}(s)$ is a set of finite words
  over a countable alphabet, and as such is countable.
  In particular, if $\Cyl\subseteq \FPaths_{\MC}(s)$ is a set of
  prefixes forming disjoint cylinders,\footnote{
  $\IPaths_{\MC}(\rho)$ and $\IPaths_{\MC}(\rho')$ share a path
  if and only if either $\rho$ is a prefix of $\rho'$
  or $\rho'$ is a prefix of $\rho$.}
  then $\mu_{\MC}(\bigcup_{\rho\in \Cyl} \IPaths_{\MC}(\rho)) =
  \sum_{\rho\in\Cyl} \mu_{\MC}(\IPaths_{\MC}(\rho))$.
  Moreover, if $\Pi\subseteq \IPaths_{\MC}(s)$ is the complement
  of a measurable set $\Pi'$, $\mu_{\MC}(\Pi) = 1-\mu_\MC(\Pi')$.

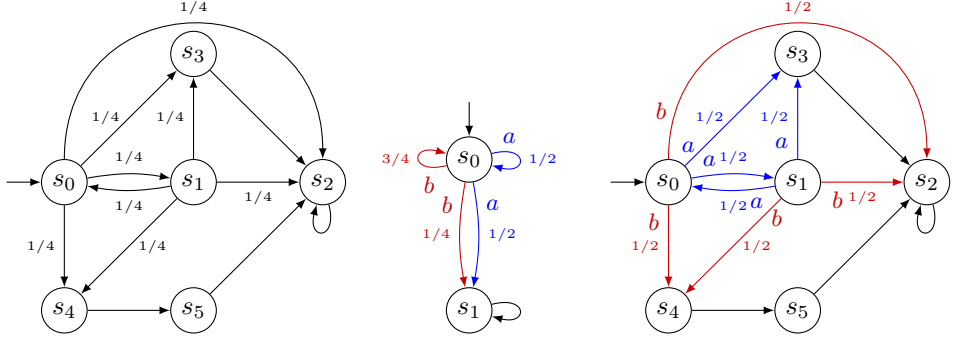
\begin{figure}[t]
\centering
\scalebox{1}{
\begin{tikzpicture}[node distance=1.7cm,auto]

\node[player0, initial](0){$s_0$};
\node[player0](1)[right of=0]{$s_1$};
\node[player0](2)[right of=1]{$s_2$};
\node[player0](3)[above of=1]{$s_3$};
\node[player0](4)[below of=0]{$s_4$};
\node[player0](5)[below of=1]{$s_5$};

\path (0) edge[action0, bend left=10] node[above]{\tiny$1/4$} (1);
\path (0) edge[action0] node[left]{\tiny$1/4$} (3);
\path (0) edge[action0, bend left=90, looseness=1.8]  node[above]{\tiny$1/4$} (2);
\path (0) edge[action0] node[left]{\tiny$1/4$} (4);

\path (1) edge[action0, bend left=10] node[below]{\tiny$1/4$} (0);
\path (1) edge[action0] node[left]{\tiny$1/4$} (3);
\path (1) edge[action0] node[below]{\tiny$1/4$} (2);
\path (1) edge[action0] node[right]{\tiny$1/4$} (4);

\path (2) edge[action0, loop below] (2);
\path (3) edge[action0] (2);
\path (4) edge[action0] (5);
\path (5) edge[action0] (2);

\end{tikzpicture}
}
\scalebox{1}{
\begin{tikzpicture}[node distance=2cm,auto]

\node[player0, initial, initial where=above](0){$s_0$};
\node[player0](1)[below of=0]{$s_1$};

\path (0) edge[action1, bend left=10] node[pos=0.2,right]{\small$a$} node[right]{\tiny$1/2$} (1);
\path (0) edge[action1, loop right] node[pos=0.2,above]{\small$a$} node[right]{\tiny$1/2$} (0);

\path (0) edge[action2, bend right=10] node[pos=0.2,left]{\small$b$} node[left]{\tiny$1/4$} (1);
\path (0) edge[action2, loop left] node[pos=0.2,below]{\small$b$} node[left]{\tiny$3/4$} (0);

\path (1) edge[action0, loop right] (1);

\end{tikzpicture}
}
\scalebox{1}{
\begin{tikzpicture}[node distance=1.7cm,auto]

\node[player0, initial](0){$s_0$};
\node[player0](1)[right of=0]{$s_1$};
\node[player0](2)[right of=1]{$s_2$};
\node[player0](3)[above of=1]{$s_3$};
\node[player0](4)[below of=0]{$s_4$};
\node[player0](5)[below of=1]{$s_5$};

\path (0) edge[action1, bend left=10] node[pos=0.2,above]{\small$a$} node[above]{\tiny$1/2$} (1);
\path (0) edge[action1] node[pos=0.2,left]{\small$a$} node[left]{\tiny$1/2$} (3);

\path (0) edge[action2, bend left=90, looseness=1.8] node[pos=0.1,left]{\small$b$} node[above]{\tiny$1/2$} (2);
\path (0) edge[action2] node[pos=0.2,left]{\small$b$} node[left]{\tiny$1/2$} (4);

\path (1) edge[action1, bend left=10] node[pos=0.2,below]{\small$a$} node[below]{\tiny$1/2$} (0);
\path (1) edge[action1] node[pos=0.2,left]{\small$a$} node[left]{\tiny$1/2$} (3);

\path (1) edge[action2] node[pos=0.2,below]{\small$b$} node[below]{\tiny$1/2$} (2);
\path (1) edge[action2] node[pos=0.2,right]{\small$b$} node[right]{\tiny$1/2$} (4);

\path (2) edge[action0, loop below] (2);

\path (3) edge[action0] (2);

\path (4) edge[action0] (5);

\path (5) edge[action0] (2);

\end{tikzpicture}
}

\caption{An MC and two MDPs, with states $s_i$ and actions $\{a,b\}$.
  In MC{s}, transitions are labelled by their probability,
  and the probability is $1$ if unspecified.
  In MDP{s}, transitions are labelled by their action and probability.
  If the action is unspecified (black transitions), then every action
  allows this transition.
  The initial state is $s_0$.
  In our examples, the set of atomic propositions is the set of states,
  so that the proposition $s_i$ holds on state $s_i$ only.
}
\label{fig:MDP}
\end{figure}

\begin{definition}
  A \emph{Markov decision process} (MDP) is a tuple
  $\MDP=\zug{ S,A,\initState, \proba, \AP, L}$,
  where $S$ is a finite set of states, $A$ is a finite set of actions,
  $\initState \in S$ is an initial state,
  $\proba : S \times A \rightarrow \Dist(S)$
  is a transition function\footnote{%
  This formalism implies that every action is available from every state.
  This is w.l.o.g., as one can model illegal actions by sending them
  to a special state.}, %
  $\AP$ is a non-empty finite set of atomic propositions,
  and $L:S \rightarrow 2^{\AP}$ is a labelling function.
\end{definition}
  If $\proba$ maps a state $s$ and an action $a$ to a distribution $d$
  so that $d(s')>0$, we write $s\xrightarrow{a,d(s')}s'$ or simply
  $s\xrightarrow{a} s'$, and we denote $\proba(s,a,s')$ the probability $d(s')$.
  We extend from MCs to MDPs the definitions and notations of finite
  and infinite paths, now labelled by actions and denoted
  $\rho = s_0\xrightarrow{a_0} s_1 \xrightarrow{a_1}\cdots$.
  Moreover, for a finite path $\rho$, we denote by $\rho\cdot a s$
  (resp. $s a\cdot\rho$) the concatenation of $\rho$ with
  $\last(\rho)\xrightarrow{a}s$
  (resp. of $s\xrightarrow{a}\first(\rho)$ with $\rho$).

  We say that $\MDP$ is stored in size $|\MDP|$ if the number of states $|S|$,
  the number of actions $|A|$ and the number of transitions
  $s\xrightarrow{a} s'$ in $\MDP$ are bounded by $|\MDP|$.
  Then, $|\FPaths_\MDP^{\leq i}|$, the number of paths of horizon at most $i$,
  is in $|\MDP|^{\mathcal O(i)}$.
  Moreover, the probabilities in $\proba$ are assumed to be rational numbers stored as
  pairs of integers $\frac{a}{b}$ in binary, so that $a,b\leq 2^{|\MDP|}$.

  A (probabilistic) \emph{strategy} is a function
  $\strat : \FPaths_{\MDP} \to \Dist(A)$
  that maps finite paths $\rho$ to distributions on actions.
  A strategy $\strat$ is \emph{deterministic} if the support of
  the distribution $\strat(\rho)$ has size $1$ for every $\rho$,
  it is \emph{memoryless} if $\strat(\rho)$ depends only on the last state
  of $\rho$, \textit{i.e.}~if $\strat$ satisfies that
  for all $\rho,\rho'\in\FPaths_{\MDP}$,
  $\last(\rho)=\last(\rho')$ implies $\strat(\rho)=\strat(\rho')$.
  We denote by $\strat(\rho,a)$ the probability of the action $a$
  in the distribution $\strat(\rho)$.

  An MDP $\MDP=\zug{S, A, \initState, \proba, \AP, L}$ equipped with
  a strategy $\strat$ defines an MC,
  denoted $\MDPstrat$, obtained intuitively by unfolding $\MDP$ and
  using $\strat$ to define the transition probabilities.
  Formally $\MDPstrat=\zug{\FPaths_{\MDP}, \initState, \probastrat, \AP, L'}$,
  with finite paths of $\MDP$ as states, transitions defined
  for all $\rho\in\FPaths_{\MDP}$, $a\in A$ and $s\in S$ by
  $\probastrat(\rho,\rho \cdot a s)=\strat(\rho,a)\proba(\last(\rho),a,s')$,
  and atomic propositions assigned by $L'(\rho)=L(\last(\rho))$.
  In particular, note that since $S$ is finite
  $\FPaths_{\MDP}$ is infinite but countable.
  We say that a finite path $\rho$ in $\MDP$ \emph{matches} a finite path $\rho'$
  in $\MDPstrat$ if $\last(\rho')=\rho$,
  so that they follow the same sequence of states and actions.
  We say that a path $\rho\in\FPaths_\MDP$ has probability $m$ in $\MDPstrat$
  if $\rho$ matches $\rho'\in\FPaths_{\MDPstrat}$ and $m$ is the measure
  of $\IPaths_{\MDPstrat}(\rho')$.
  It corresponds to the likelihood of having $\rho$ as a prefix when
  following $\strat$ and starting from $\first(\rho)$.

  We may omit $\MC$ or $\MDP$ from all previous notations
  when they are clear from the context. MC notations may use $\strat$
  as shorthand for $\MDPstrat$, \textit{e.g.}~$\mu_{\strat}$ is the probability
  measure induced by $\MDPstrat$, and $\FPaths_{\strat}$ refers to finite paths
  of non-zero probability under $\strat$.

\begin{example}\label{ex:MDP}
  Consider the MC on the left of \figurename~\ref{fig:MDP},
  and the property asking to reach the state $s_2$ in at most two steps.
  Consider the set of paths of length at most two from $s_0$ to $s_2$.
  Let $\Pi=\IPaths(s_0s_2) \uplus \IPaths(s_0s_1s_2) \uplus \IPaths(s_0s_3s_2)$
  be the infinite paths obtained from their cylinders.
  Then, the probability of reaching $s_2$ in two steps
  when starting from $s_0$ is
  $\mu(\Pi)=\frac{1}{4}+\frac{1}{16}+\frac{1}{4}=\frac{9}{16}$.
  Note that the probability of reaching $s_2$ in two steps
  when starting from $s_1$ is also
  $\frac{9}{16}$. Every other state reaches $s_2$ with probability $1$
  in two steps. Consider now the MDP in the middle of \figurename~\ref{fig:MDP},
  and the property asking that the state reached after the first transition
  is $s_1$. For every strategy $\strat$, the probability that this property
  holds in $\MDPstrat$ is equal to $\strat(s_0,a)\frac{1}{2}
  +\strat(s_0,b)\frac{1}{4}=\strat(s_0,a)\frac{1}{2}
  +(1-\strat(s_0,a))\frac{1}{4}=\strat(s_0,a)\frac{1}{4}+\frac{1}{4}$.
\end{example}

\subparagraph*{{\bf Probabilistic CTL with Linear expressions}}
  A formula of $L$-PCTL is generated by the nonterminal $\Phi$
  in the following grammar:
\begin{definition}[$L$-PCTL in normal form, syntax]\label{def:normal_form}
  \begin{align*}
    \Phi &:= p \mid \neg p \mid \Phi_1 \wedge \Phi_2 \mid \Phi_1 \vee \Phi_2
    \mid \sum_{i=1}^n c_i \Pr{\varphi_i} \succcurlyeq c_0 \\
    \varphi &:= \XX^{\ell} \Phi \mid \Phi_1 \UU^{\ell} \Phi_2 \mid
      \Phi_1 \WW^{\ell} \Phi_2 \mid \Phi_1 \UU^{\infty} \Phi_2 \mid
      \Phi_1 \WW^{\infty} \Phi_2
  \end{align*}
  where $p$ ranges over the atomic propositions in $\AP$,
  $\ell$ ranges over $\N$,
  and $n\in\Nspos$, $(c_0,\cdots,c_n)\in\Z^n$, %
  ${\succcurlyeq}\in\{{\geq},{>}\}$
  define linear inequalities.
\end{definition}

  We call a formula generated by $\Phi$ a \emph{state formula},
  and a formula generated by $\varphi$ a \emph{path formula}.
  The \emph{horizon label} of a path formula is the label of its root operator,
  \textit{i.e.}~either $\ell$ or $\infty$.
  Intuitively, the Next operator $\XX^{\ell} \Phi$ means that $\Phi$ holds
  in exactly $\ell$ steps, the (unbounded) Until and Weak until operators
  $\UU^{\infty}$ and $\WW^{\infty}$ are defined as usual in CTL,
  and their bounded version $\Phi_1 \UU^{\ell} \Phi_2$ and
  $\Phi_1 \WW^{\ell} \Phi_2$ impose a horizon on the reachability of $\Phi_2$.
  We will use the standard notations $\XX$, $\UU$ and $\WW$, defined by
  $\XX^1$, $\UU^{\infty}$ and $\WW^{\infty}$, respectively.

\begin{definition}[$L$-PCTL in normal form, semantics]\label{def:semantics}
  For a fixed MC $\MC$ of states~$S$,
  we inductively define $\sem{\Phi}_\MC$ as a set of states, and
  for each state $s$ we define $\sem{\varphi}_\MC^s$ as a measurable set
  of infinite paths starting from $s$:
  \begin{align*}
    \sem{p}_\MC &= \{s\in S \mid p \in L(s)\}\qquad\quad
    \sem{\neg p}_\MC = \{s\in S \mid p \not\in L(s)\}\\
    \sem{\Phi_1 \wedge \Phi_2}_\MC &= \sem{\Phi_1}_\MC\cap\sem{\Phi_2}_\MC\qquad
    \sem{\Phi_1 \vee \Phi_2}_\MC = \sem{\Phi_1}_\MC\cup\sem{\Phi_2}_\MC\\
    \sem{\sum_{i=1}^n c_i \Pr{\varphi_i} \succcurlyeq c_0}_\MC &= \{s\in S \mid
      \sum_{i=1}^n c_i~\mu_\MC(\sem{\varphi_i}_\MC^s) \succcurlyeq c_0\}\\
    \sem{\XX^{\ell} \Phi}_\MC^s &= \{\rho\in\IPaths(s) \mid
      \rho[\ell] \in \sem{\Phi}_\MC\}\\
    \sem{\Phi_1 \UU^{\ell} \Phi_2}_\MC^s &= \{\rho\in\IPaths(s) \mid
      \exists j\leq \ell, \rho[j] \in \sem{\Phi_2}_\MC \wedge
      \forall i<j, \rho[i] \in \sem{\Phi_1}_\MC\}\\
    \sem{\Phi_1 \WW^{\ell} \Phi_2}_\MC^s &= \{\rho\in\IPaths(s) \mid
      \forall j\leq \ell, (\rho[j] \in \sem{\Phi_1}_\MC
      \vee \rho[j] \in \sem{\Phi_2}_\MC) \\
      &\qquad\qquad\qquad\vee \exists i<j, \rho[i] \in \sem{\Phi_2}_\MC\}\\
    \sem{\Phi_1 \UU \Phi_2}_\MC^s &= \{\rho\in\IPaths(s) \mid
      \exists j\in\N, \rho[j] \in \sem{\Phi_2}_\MC \wedge
      \forall i<j, \rho[i] \in \sem{\Phi_1}_\MC\}\\
    \sem{\Phi_1 \WW \Phi_2}_\MC^s &= \{\rho\in\IPaths(s) \mid
      \forall j\in\N, (\rho[j] \in \sem{\Phi_1}_\MC
      \vee \rho[j] \in \sem{\Phi_2}_\MC) \\
      &\qquad\qquad\qquad\vee \exists i<j, \rho[i] \in \sem{\Phi_2}_\MC\}
  \end{align*}
\end{definition}
  Then, we write $s\models_\MC \Phi$ (resp. $\rho\models_\MC\varphi$)
  if $s\in\sem{\Phi}_\MC$ (resp. $\rho\in\sem{\varphi}_\MC^{\first(\rho)}$),
  and say that $s$ satisfies $\Phi$ (resp. $\rho$ satisfies $\varphi$).
  We denote by $\equiv$ the semantic equivalence of state or path formulae
  (that holds on all MCs).
  Finally, we write $\MC \models \Phi$ if $\initState\models_\MC \Phi$.
  Note that by restricting the linear inequalities to $n=1$
  and
$\ell=1$ in $\XX^\ell$,
  we recover the standard definition of PCTL
  (see \textit{e.g.}~\cite{BK08}).

  We define usual notions as syntactic sugar,
  so that state formulae allow for
  $\bot := p \wedge \neg p$ and $\top := p \vee \neg p$ (for any $p\in \AP$).
  We allow rational constants $c_1$ and all comparison symbols
  in $\{\leq,<, =, \neq,>,\geq\}$ in linear expressions, with
  $\sum_{i=1}^n c_i \Pr{\varphi_i} \preccurlyeq c_0 :=
  \sum_{i=1}^n (-c_i) \Pr{\varphi_i} \succcurlyeq -c_0$
  and $=,\neq$ defined as conjunctions or disjunctions.
  Moreover, path formulae allow for
  $\FF^\ell \Phi := \top \UU^\ell \Phi$ and
  $\GG^\ell \Phi := \Phi \WW^\ell \bot$.
  We allow the negation operation $\neg$ in state and path formulae,
  and recover a formula in normal form using De Morgan's laws, the negation
  of inequalities ($\geq$ becomes $<$ and $>$ becomes $\leq$),
  and the duality rule $\neg(\Phi_1 \WW^{\ell} \Phi_2)
  \equiv (\neg\Phi_1 \wedge \neg \Phi_2) \UU^{\ell} \neg\Phi_1$.
  Finally, boolean implication and equivalence are defined as usual.
  A notable property is $\Phi_1 \WW^{\ell} \Phi_2
  \equiv (\Phi_1 \UU^{\ell} \Phi_2) \vee \GG^\ell \Phi_1$.

  We encode $L$-PCTL formulae as trees,
  whose internal nodes are labelled by state or path operators
  and whose leaves are labelled by atomic propositions.
  Let $\ellmax\geq 1$ denote an upper bound on horizon labels $\ell$
  of subformula of $\Phi$ where $\ell$ is finite.
  The constants $c_i$ in linear expressions are encoded in binary,
  and \emph{the horizon labels $\ell$ are encoded in unary},
  so that if the overall encoding of $\Phi$ is of size $|\Phi|$,
  we shall have $\ellmax\leq |\Phi|$.
  We argue that this choice is justified from a larger point of view that extends
  PCTL to PCTL$^*$ by allowing boolean operations in path formulae,
  as the bounded horizon operators $\XX^\ell, \UU^\ell, \WW^\ell$ can be seen
  as syntactic sugar for a disjunction of nested $\XX$ operators
  of size $\mathcal O(\ell)$.

\begin{definition}
  An $L$-PCTL formula $\Phi$ (in normal form) is a \emph{window formula} if
  the horizon label $\ell$ of every path operator in $\Phi$
  is finite, so that the unbounded $\UU$ and $\WW$ are not used.
  It is a \emph{non-strict} formula if $\succcurlyeq$ is always $\geq$
  in its linear inequalities.
  It is a \emph{flat} formula if the measure operator $\proba$ is never nested,
  so that if $\Phi$ is seen as a tree, every branch has at most one node
  labelled by a linear inequality $\sum_{i=1}^n c_i \Pr{\varphi_i} \succcurlyeq c_0$.
\end{definition}

\begin{definition}
  A \emph{global window} formula is
  a formula of the shape $\Forall \GG \Phi$,
  with $\Phi$ a window $L$-PCTL formula.
  It is satisfied by a state $s$ of $\MC$ if
  every infinite path in $\IPaths_\MC(s)$ satisfies
  the path formula $\GG \Phi$, or equivalently if every state reachable from $s$
  satisfies $\Phi$.
\end{definition}

\begin{lemma}\label{lm:global-window-def}
  The \emph{global window} formula $\Forall \GG \Phi$ is satisfied
  on a state $s$ of $\MC$ if and only if $s$ satisfies
  the $L$-PCTL formula $\Pr{\GG \Phi} = 1$.
\end{lemma}
\begin{proof}%
  If $\Forall \GG \Phi$ holds on $s$, then
  $\mu_\MC(\sem{\GG \Phi}^s_\MC)=\mu_\MC(\IPaths_\MC(s))=1$.
  If $\Forall \GG \Phi$ is not satisfied on $s$, then
  there exists a finite path $\rho$ leading to a state that violates $\Phi$,
  so that the entire cylinder $\IPaths_\MC(\rho)$ satisfies
  the path formula $\FF \neg\Phi$. It follows that
  $\mu_\MC(\sem{\GG \Phi}^s_\MC)=
  1-\mu_\MC(\sem{\FF \neg\Phi}^s_\MC)\leq
  1-\mu_\MC(\IPaths_\MC(\rho))<1$.
\end{proof}

\begin{example}\label{ex:PCTL}
  Consider the MC $\MC$ to the left of \figurename~\ref{fig:MDP}.
  Let $\Phi$ be the $L$-PCTL formula $\Pr{\FF^2 s_2}\geq\frac{9}{16}$.
  It is a window formula, that is flat and non-strict. As detailed
  in Example~\ref{ex:MDP}, every state of $\MC$ satisfies $\Phi$.
  Therefore, $\MC$ satisfies the global window formula $\Forall \GG \Phi$.
  Consider now the MDP $\MDP$ to the right of \figurename~\ref{fig:MDP}.
  Let $\strat$ denote the memoryless strategy that chooses,
  in $s_0$ and $s_1$, action $a$
  with probability $\frac{1}{2}$ and action $b$ with probability $\frac{1}{2}$.
  While $\MDPstrat$ is an infinite MC by definition,
  it is bisimilar to the MC on the left of \figurename~\ref{fig:MDP}
  and must satisfy the same PCTL formulae,
  so that $\MDPstrat \models \Forall \GG \Phi$.
  In Section~\ref{sec:fixpoint}, we will show that
  $\strat$ is the only strategy on $\MDP$ that satisfies $\Forall \GG \Phi$.
\end{example}

\subparagraph*{{\bf Model checking and synthesis problems}}
  The \emph{model-checking problem} of an $L$-PCTL formula $\Phi$ and of
  a finite MC $\MC$ is the decision problem
  asking if $\MC \models \Phi$.
  The \emph{synthesis problem} of an $L$-PCTL formula $\Phi$ and of
  an MDP $\MDP$
  asks if there exists a strategy $\strat$ so that $\MDPstrat \models \Phi$.
  We also consider the sub-problems that restrict the set of strategies
  to subsets defined by constraints on the memory or on determinism.
  For example, the memoryless (resp.~deterministic) synthesis problem asks for
  a memoryless (resp.~deterministic) strategy satisfying the formula.
  They are indeed distinct problems:
\begin{example} \label{example:memory}
  Consider the MDP in the middle of \figurename~\ref{fig:MDP}.
  Let $\Phi$ be the window formula
  $(\Pr{\FF^2 s_1}=\frac{5}{8})\wedge (\Pr{\XX s_1}\geq\frac{1}{2}
  \vee \Pr{\XX s_1}\leq\frac{1}{4})$.
  First, $s_0\models \Pr{\XX s_1}\geq\frac{1}{2} \Leftrightarrow
  \strat(s_0, a)=1$ and $s_0\models \Pr{\XX s_1}\leq\frac{1}{4} \Leftrightarrow
  \strat(s_0, a)=0$, so that the first move must be deterministic.
  If the first action is $a$, and the transition $s_0\xrightarrow{a} s_0$
  is chosen, then the next choice must be $b$ to ensure
  $\Pr{\FF^2 s_1}=\frac{5}{8}$. Similarly, if the first action is $b$,
  the next choice on $s_0$ must be~$a$.
  Moreover, $s_1 \models \Phi$ under any strategy.
  Thus, the only strategies that satisfy $\Forall \GG \Phi$ are the
  strategies that alternate between $a$ and $b$ as long as we are in $s_0$,
  while no memoryless strategy satisfies $\Forall \GG \Phi$.
  On the other hand, in Example~\ref{ex:PCTL} randomisation
  is needed. An example that require both randomisation and memory can be constructed
  by combining both examples.
\end{example}

\begin{proposition}\label{prop:model-checking}
  The \emph{model-checking} problem for $L$-PCTL formulae
  and finite MCs
  can be solved in $\P$.
  This comes at no extra cost when compared to standard PCTL~\cite{BK08}.
\end{proposition}
\begin{proof}%
  This problem is detailed in~\cite[Thm. 10.40]{BK08} for a PCTL definition
  that only allows expressions of the shape
  $c_1 \Pr{\varphi} \succcurlyeq c_0$ to quantify over path formulae.
  Extending to $\sum_{i=1}^n c_i \Pr{\varphi_i} \succcurlyeq c_0$ is
  straight-forward, as their algorithm computes the measure of $\Pr{\varphi}$,
  and then checks if the comparison holds.
  The intuition is that the measure of bounded operators $\XX^\ell$, $\UU^\ell$
  and $\WW^\ell$ are obtained by $\mathcal O(\ell)$
  vector-matrix multiplications, while unbounded $\UU$
  and $\WW$ are seen as linear equation systems.
  Overall, the complexity is in $|\MC|^{\mathcal O(1)} |\Phi| \ellmax$.
\end{proof}

\section{Synthesis for global window PCTL}\label{sec:fixpoint}

  In this section, we detail complexity results
  on the synthesis problem for global window $L$-PCTL formulae.
  We fix a Markov decision process $\MDP$, a formula $\Forall \GG \Phi$ where $\Phi$ is a window $L$-PCTL formula,
  and ask if there exists
  a strategy $\strat$ so that $\MDPstrat\models \Forall \GG \Phi$.
  We also address the sub-problems concerning deterministic
  or memoryless strategies.

\subparagraph*{{\bf Solving window formulae}}
  We start by constructing a strategy $\strat$ so that $\MDPstrat\models \Phi$.
  The formula $\Phi$ can be seen syntactically
  as a tree with state or path operators
  on internal nodes and atomic propositions on leaves.
  The window length of a branch of this tree is the sum
  of the horizon labels $\ell$ of path operators in the branch.
  The \emph{window length} of the formula $\Phi$ is
  an integer $\Ell$ obtained as the maximum over every branch of $\Phi$
  of their respective window lengths. In particular, $\Ell\leq |\Phi|\ellmax$.
  For example, if $\Phi=\Pr{\XX \Pr{\XX^2 p_1}\geq \frac{1}{2}}\leq \frac{1}{2}
  \vee \Pr{\FF^2 p_3}>0$, then $\ellmax=2$ and $\Ell=\max(1+2,2)=3$.

\begin{definition} \label{def:winstrat}
  Let $s$ be a state of $\MDP$
  and $\Phi$ be a window $L$-PCTL formula of window length $\Ell$.
  A \emph{window strategy} for $s$ of horizon $\Ell$
  is a mapping $\stratW: \FPaths_{\MDP}^{<\Ell}(s) \to \Dist(A)$.
\end{definition}

  A window strategy $\stratW$ for state $s$ can be seen as a partial strategy,
  only defined on paths of length under $\Ell$ that start from $s$.
  Formally, $\stratW$ defines
  a set of strategies $\strat: \FPaths_{\MDP} \to \Dist(A)$,
  where the first $\Ell$ steps from $s$
  are specified by $\stratW$, and the subsequent steps are not.
  This set of strategies is called the \emph{cylinder} of
  the window strategy $\stratW$.
  In particular, if two strategies $\strat$ and $\strat'$ are in the cylinder of
  the window strategy $\stratW$, then the MCs $\MDPstrat$
  and $\MDPtoMC{\MDP}{\strat'}$ coincide for the first $\Ell$ steps from $s$,
  in the sense that every path $\rho\in\FPaths_\MDP^{<\Ell}(s)$
  has the same probability $m$
  in $\MDPstrat$ and in $\MDPtoMC{\MDP}{\strat'}$.
  In this case, we say that $m$ is the probability of $\rho$ under
  $\stratW$.

  We may conflate a window strategy $\stratW$ with an arbitrary
  strategy $\strat$ in its cylinder, so that
  $\FPaths^{<\Ell}_{\stratW}(s)$ is a set of paths in $\MDPstrat$.
  Then, we say that the window strategy $\stratW$ for state $s$
  satisfies $\Phi$,
  noted $s\models_{\stratW} \Phi$, if $s\models_{\strat} \Phi$
  for all $\strat$ in the cylinder of $\stratW$.
  Conversely, a strategy $\strat: \FPaths_{\MDP} \to \Dist(A)$ naturally defines
  a window strategy $\stratW_\rho$ for every fixed prefix $\rho$, so that
  for all $\rho'\in\FPaths_{\MDP}^{<\Ell}(\last(\rho))$,
  $\stratW_\rho(\rho')=\strat(\rho\cdot\rho')$.

\begin{lemma}\label{lm:window-length}
  Let $\Phi$ be a window $L$-PCTL formula of window length
  $\Ell$, $\strat$ be a strategy for $\MDP$,
  and let $\stratW_s$ be the window strategy defined by $\strat$ on
  state $s$ and horizon $\Ell$ (the fixed prefix is $s$). Then, it holds that
  $s\models_{\strat}\Phi \Leftrightarrow
  s\models_{\stratW_s}\Phi$.
\end{lemma}

  Thus, the synthesis problem on window formulae reduces to finding
  a window strategy $\stratW$ for $\initState$
  so that $\initState \models_{\stratW} \Phi$.
  Let $\stratW$ be a window strategy for state $s$ and horizon $\Ell$.
  Let $\mathcal X_s$ denote a finite set of variables $x_{\rho,a}$,
  with $\rho\in\FPaths_\MDP^{<\Ell}(s)$, and $a\in A$.
  The window strategy $\stratW$ can be seen as a point in
  the real number space $\R^{\mathcal X_s}$,
  where $x_{\rho,a}$ encodes $\stratW(\rho,a)$.
  Conversely, every point in $\R^{\mathcal X_s}$ so that
  $\forall x\in \mathcal X_s$, we have
  $x\in[0,1]$, and $\forall \rho\in \FPaths_\MDP^{<\Ell}(s)$, we have
  $\sum_{a\in A} x_{\rho,a} = 1$
  represents a window strategy.
  Therefore, the points of $\R^{\mathcal X_s}$ that encode
  a window strategy can be described
  by a finite conjunction of linear inequalities $x\geq 0$, $x\leq 1$
  and $x_{\rho,a_1}+\cdots+ x_{\rho,a_k}=1$ over the variables $\mathcal X_s$.

  We want to similarly characterise the set of window strategies
  satisfying a given window $L$-PCTL formula.
  As will become apparent later on, we will need polynomial inequalities.

\begin{definition}
  The \emph{first-order theory of the reals} (FO-\R) is the set of all
  well-formed sentences of first-order logic that involve universal and
  existential quantifiers and logical combinations of equalities
  and inequalities of real polynomials.\footnote{%
  The primitives operations are multiplication and addition,
  the comparison symbols are $\{\leq,=,\geq\}$.}
\end{definition}
  We allow the use of strict comparison operators $\{<,\neq,>\}$ as the negation
  of their non-strict versions.
  We also assume that
  the formula is written in \emph{prenex normal form} (PNF),
  \textit{i.e.}~as a sequence of alternating blocks
  of quantifiers followed by a quantifier-free formula.
  Finally, if $\mathcal S=\{x_1,\cdots,x_k\}$ is a finite set of variables,
  we use the notation $\operatorname*{\exists} \mathcal S$
  as shorthand for the quantifier sequence $\exists x_1\cdots\exists x_k$.

  This theory is decidable,
  and admits a doubly-exponential quantifier elimination
  procedure~\cite{RENEGAR1992329}.
  Of particular interest is the existential fragment of FO-$\R$,
  denoted $\exists$-$\R$, where only $\exists$ is allowed.
  It is decidable in \PSPACE~\cite{10.1145/62212.62257}.

  We say that an FO-$\R$ formula of free variables $\mathcal X$
  is \emph{non-strict} if it is satisfied by a closed set of points
  in $\R^{\mathcal X}$. In particular, an FO-$\R$ formula that only uses
  non-strict comparison symbols $\{\leq,=,\geq\}$ and
  that is negation-free\footnote{%
  A formula is negation-free if the Boolean negation operator $\neg$ is
  not used.}
  is non-strict.

\begin{proposition}\label{prop:window-FOR}
  Let $s$ be a state of $\MDP$
  and $\Phi$ be a window $L$-PCTL formula.
  The set of window strategies $\stratW$ such that $s\models_{\stratW} \Phi$
  can be represented in $\exists$-$\R$ as a PNF formula
  of free variables~$\mathcal X_s$.
  This formula is of size $|\Phi||\MDP|^{\mathcal O(\Ell)}$,
  and can be computed in \EXP.
  If $\Phi$ is flat and non-strict then the $\exists$-$\R$ formula
  is non-strict.
\end{proposition}

\begin{proof}[Proof sketch]
  We encode the problem in the theory of the reals, by using free variables
  $x_{\rho,a}\in[0,1]$ that have the value of $\stratW(\rho,a)$,
  existential variables $y_{\rho,\Phi'}\in\{0,1\}$ that are true if the state subformula $\Phi'$ is satisfied
  when one follows $\stratW$ after a fixed history of $\rho$,
  and existential variables  $z_{\rho,\varphi}\in[0,1]$ having the probability
  that the path subformula $\varphi$ is satisfied
  when one follows $\stratW$ after a fixed history of $\rho$.
  The $z$ variables use \enquote{local consistency equations} that equate the probability
  of a path formula on the current state as a linear combination of its probability
  on the
  successor states.
  For the $\XX^\ell \Phi'$ formula this translates into $z_{\rho,\XX^\ell \Phi'}=\sum_{s\xrightarrow{a}s'} x_{\rho,a} \proba(s,a,s')z_{\rho\cdot as',\XX^{\ell-1} \Phi'}$ for example.
  The $y$ variables can then be defined, so that
  $y_{\rho,(\sum_{i=1}^n c_i \Pr{\varphi_i} \succcurlyeq c_0)} = 1$
  if and only if $\sum_{i=1}^n c_i z_{\rho,\varphi_i} \succcurlyeq c_0$.
  Lastly we ask that $y_{s,\Phi}=1$.
  To maintain the non-strict property, some subtlety is needed
  in the way nested probabilistic operators are dealt with.
\end{proof}

  If we use a $\PSPACE$ decision procedure for $\exists$-$\R$
  on the formula of Proposition~\ref{prop:window-FOR}, we get:
\begin{theorem}\label{thm:window-synthesis}
  The synthesis problem for \emph{window} $L$-PCTL formulae
  is in $\EXPSPACE$.
\end{theorem}

\begin{example}\label{ex:window-FOR}
  Consider the MDP $\MDP$ to the right of \figurename~\ref{fig:MDP},
  and $\Phi=\Pr{\FF^2 s_2}\geq\frac{9}{16}$.
  Let us describe the formula %
  obtained by Proposition~\ref{prop:window-FOR}.
  $s_0\models_{\partial} \Phi$ can be encoded schematically
  as the formula
  $\exists z_{s_0,\FF^2 s_2}\exists z_{s_0as_1,\FF^1 s_2}, \text{s.t.}~
  z_{s_0,\FF^2 s_2} =
  \frac{x_{s_0,a}}{2}z_{s_0as_1,\FF^1 s_2}
  +\frac{x_{s_0,a}}{2}+\frac{x_{s_0,b}}{2}
  \wedge z_{s_0as_1,\FF^1 s_2} = \frac{x_{s_0as_1,b}}{2}
  \wedge z_{s_0,\FF^2 s_2}\geq\frac{9}{16}$.
  For readability reasons,
  we simplified boolean expressions
  involving $\top$ or $\bot$ when appropriate, and we omitted the
  variables that can be simplified out immediately, as well as the
  constraints making sure that the variables $x$ encode probabilities.

  After quantifier elimination, and using $x_{s_0,a}+x_{s_0,b}=1$,
  we get $x_{s_0,a}\ x_{s_0as_1,b}\geq \frac{1}{4}$.
  Observe that, as mentioned in Example~\ref{ex:PCTL},
  a window strategy $\stratW$ that sets $x_{s_0,a}=x_{s_0as_1,b}=\frac{1}{2}$
  satisfies $\Phi$.
  Similarly, $s_1\models_{\partial} \Phi$ can be encoded as
  $x_{s_1,a}\ x_{s_1as_0,b}\geq \frac{1}{4}$.
\end{example}

\subparagraph*{{\bf Fixed point characterisation of global window formulae}}
  Let $\Phi$ be a window $L$-PCTL formula of window length $\Ell$.
  In this subsection, we describe a fixed point characterisation of
  the synthesis problem for the global window formula $\Forall \GG \Phi$.

  A \emph{window strategy portfolio} $\Pi$
  of horizon $\Ell$ (in short, a portfolio $\Pi$) maps each state $s$
  to a set $\Pi_s$ of window strategies
  for $s$ of horizon $\Ell$. A window strategy portfolio can be seen as a set
  of points in $\R^{\mathcal X_s}$ for every state $s$.
  Given two window strategy portfolios $\Pi$ and $\Pi'$ of horizon $\Ell$,
  we write $\Pi\subseteq\Pi'$ if for all $s\in S$,
  it holds that $\Pi_s\subseteq\Pi_s'$.
  Then, the set of all window strategy portfolios of horizon $\Ell$
  is a complete lattice w.r.t.~$\subseteq$,
  where for a set $\mathcal S$ of portfolios, the meet
  $\bigsqcap \mathcal S$
  (resp.
  the join $\bigsqcup \mathcal S$)
  maps $s$ to $\bigcap_{\Pi\in \mathcal S} \Pi_s$
  (resp. $\bigcup_{\Pi\in \mathcal S} \Pi_s$).

  Let $s\xrightarrow{a}s'$ be a transition in $\MDP$, and
  let $\stratW$, $\stratW'$ be window strategies for $s$ and $s'$, respectively,
  of horizon $\Ell$. We say that $\stratW$ and $\stratW'$ are \emph{compatible}
  w.r.t.~$s\xrightarrow{a}s'$ if they make the same decisions on shared paths,
  \textit{i.e.}~for all $\rho\in\FPaths_{\MDP}^{<\Ell}(s')$
  the probability of $sa\cdot \rho$ under $\stratW$ equals
  the probability of $\rho$ under $\stratW'$
  multiplied by $\stratW(s,a)\proba(s,a,s')$.
  In particular, whenever $sa\cdot \rho$ has non-zero probability
  under $\stratW$ and $|\rho|<\Ell-1$,
  we have $\stratW(sa\cdot \rho) = \stratW'(\rho)$.
  Similarly, we say that $\stratW$ and a set $\Pi_{s'}$
  of window strategies for $s'$ are compatible w.r.t.~$s\xrightarrow{a}s'$
  if either $\stratW(s,a)=0$ or there exists a window strategy $\stratW'$
  in $\Pi_{s'}$ so that $\stratW$ and $\stratW'$ are compatible
  w.r.t.~$s\xrightarrow{a}s'$.

  Let $\f$ map portfolios to portfolios,
  so that $\f(\Pi)$ maps $s\in S$ to the set $\f(\Pi)_s$ of window
  strategies $\stratW\in \Pi_s$ so that for each $s\xrightarrow{a}s'$ in $\MDP$,
  we have that $\stratW$ and $\Pi_{s'}$ are compatible w.r.t.~$s\xrightarrow{a}s'$.
  Intuitively, $\f$ removes from $\Pi_s$ the window strategies $\stratW$
  that are not compatible with any continuation after
  a transition $s\xrightarrow{a}s'$ for some action $a$.
  Then, $\f$ is expressible in the theory of the reals:
\begin{lemma}\label{lm:fixpoint-FOR}
  Let $\Pi$ be a portfolio, encoded as an $\exists$-$\R$ formula $R^s$,
  of free variables $\mathcal X_s$, for every state $s$.
  Assume that each $R^s$ is a PNF formula of size $F$.
  Then, $\f(\Pi)_s$ can also be encoded as a PNF formula,
  of size in $\mathcal O(|\MDP|F)+|\MDP|^{\mathcal O(\Ell)}$.
  Moreover, if the formulae associated with $\Pi$ are non-strict,
  so are the formulae of $\f(\Pi)$.
\end{lemma}
\begin{proof}%
  Let $s\xrightarrow{a}s'$ be a transition in $\MDP$, and
  let $\stratW$, $\stratW'$ be window strategies for $s$ and $s'$,
  encoded as points in $\R^{\mathcal X_s}$ and $\R^{\mathcal X_{s'}}$,
  respectively.
  If $\rho\in\FPaths_{\MDP}^{<\Ell}(s')$,
  then let $\POLY(\rho)$ denote the polynomial $\prod_{0\leq i <|\rho|}
  x_{\rho[{:}i],a_i} \proba(s_i,a_i,s_{i+1})$.
  Then, the strategies $\stratW$ and $\stratW'$ are compatible
  w.r.t.~$s\xrightarrow{a}s'$ if for all $\rho\in\FPaths_{\MDP}^{<\Ell-1}(s')$
  so that $\POLY(\rho)>0$, for all $a'\in A$
  we have $x_{sa\cdot \rho,a'} = x_{\rho, a'}$.
  Then, if $\Pi_{s'}$ is encoded as the formula $R^{s'}$,
  we get that $\stratW$ and $\Pi_{s'}$
  are compatible w.r.t.~$s\xrightarrow{a}s'$ if there exists a valuation
  of $\mathcal X_{s'}$ that encodes a window strategy $\stratW'$ so that
  $\stratW$ and $\stratW'$ are compatible w.r.t.~$s\xrightarrow{a}s'$.
  This property corresponds to the formula
  defined by
  \[\mathcal F(s,a,s') :=
    \exists \mathcal X_{s'}, R^{s'} \wedge
    \bigwedge_{\rho\in\FPaths_{\MDP}^{<\Ell-1}(s')} \POLY(\rho)=0
    \vee \bigwedge_{a'\in A} x_{sa\cdot \rho, a'} = x_{\rho, a'}\]

  Therefore, if $\Pi_s$ is encoded as $R^s$ then the formula
  $R^s \wedge \bigwedge_{s\xrightarrow{a}s'}
  x_{s,a}=0 \vee \mathcal F(s,a,s')$ encodes $\f(\Pi)_s$.
  Observe that it introduces non-strict comparisons,
  existential quantifiers, and no negation operations,
  and is of size in $F+|\MDP|(F+|\MDP|^{\mathcal O(\Ell)})$.
\end{proof}

\begin{example}\label{ex:window-fixpoint-op}
  Consider again the MDP $\MDP$ to the right of \figurename~\ref{fig:MDP}.
  Let $\Pi$ be the portfolio where $\Pi_{s_0}$
  is defined by $x_{s_0,a}\ x_{s_0as_1,b}\geq \frac{1}{4} \wedge
  x_{s_0as_1, b}\leq c \text{ with }c\in\left[\frac{1}{2},1\right]\,,$
  $\Pi_{s_1}$ is defined by
  $x_{s_1,a}\ x_{s_1as_0,b}\geq \frac{1}{4} \wedge x_{s_1as_0, b}\leq c\,,$
  and $\Pi$ is $\top$ on every other state of $\MDP$.\footnote{%
  We omitted the constraints that ensure that all variables encode probabilities.}
  Then, using Lemma~\ref{lm:fixpoint-FOR} yields that a formula equivalent to
  $x_{s_0,a}\ x_{s_0as_1,b}\geq \frac{1}{4} \wedge
  x_{s_0as_1, b}\leq 1-\frac{1}{4c}$ encodes $\f(\Pi)_{s_0}$.
  Symmetrically, $\f(\Pi)_{s_1}$ can be encoded as
  $x_{s_1,a}\ x_{s_1as_0,b}\geq \frac{1}{4}
  \wedge x_{s_1as_0, b}\leq 1-\frac{1}{4c}$.
\end{example}

\begin{lemma}[Knaster-Tarski, Kleene]\label{lm:fixed-points}
  The operator $\f$ is Scott-continuous (upwards and downwards), and is thus monotone.
  Let $Q$ be a set of window strategy portfolios of horizon $\Ell$
  that forms a complete lattice w.r.t.~$\subseteq$.
  Then, the set of fixed points of $\f$ in $Q$ forms a complete lattice
  w.r.t.~$\subseteq$.
  Moreover, $\f$ has
  a greatest fixed point in $Q$
  equal to $\bigsqcap \{\f^n(\bigsqcup Q)\mid n\in\N\}$.
\end{lemma}

  Let $\Phi$ be a window $L$-PCTL formula of window length $\Ell$.
  Let $Q^\Phi = \{\Pi \mid \forall s\in S, \forall \stratW\in \Pi_s,
  s \models_{\stratW} \Phi \}$ be the set of portfolios
  containing window strategies of horizon $\Ell$ that ensure $\Phi$.
  It is closed by $\bigsqcap$ and $\bigsqcup$, and therefore forms
  a complete lattice.
  The greatest element $\bigsqcup Q^\Phi$ is the full
  portfolio mapping every $s$ to all window strategies $\stratW$
  so that $s \models_{\stratW} \Phi$.
  We denote $\Pi^\Phi$ the greatest fixed point of $\f$ in $Q^\Phi$,
  that must exist by Lemma~\ref{lm:fixed-points}.

\begin{proposition}\label{prop:fixed-point}
  Let $s_0$ be a state, and let $\Phi$  be a window $L$-PCTL formula.
  Then, $\Pi^\Phi_{s_0}\neq \emptyset$
  if and only if there exists a strategy $\strat$ so that
  $s_0 \models_{\strat} \Forall \GG \Phi$.
\end{proposition}
\begin{proof}[Proof sketch]
  On the one hand, we show that if $\strat$ is a strategy so that
  $s_0 \models_{\strat} \Forall \GG \Phi$,
  and if $\Pi^{s_0,\strat}$ is the set of window strategies obtained
  for state $s$ and horizon $\Ell$ from fixed prefixes
  of non-zero probability in $\strat$,
  then $\Pi^{s_0,\strat}$ is a fixed point of $\f$ in $Q^\Phi$
  so that $\Pi^{s_0,\strat}_{s_0}\neq\emptyset$.
  On the other hand, we show that from every fixed point $\Pi$ of $\f$ in $Q^\Phi$
  that is non-empty on a state $s_0$, we can inductively construct a
  strategy $\sigma$ so that $s_0 \models_{\strat} \Forall \GG \Phi$,
  that intuitively consists in picking successive window strategies from $\Pi$
  that are compatible with each other.
\end{proof}

  Therefore, computing $\Pi^\Phi$ solves the synthesis problem for
  global window $L$-PCTL formulae.
  By Lemma~\ref{lm:fixed-points}, we have that $\Pi^\Phi$ is the limit
  of the non-increasing sequence $(\f^i(\bigsqcup Q^\Phi))_{i\in\N}$,
  with $\bigsqcup Q^\Phi$ being the full portfolio that can be obtained as
  an $\exists$-$\R$ formula by Proposition~\ref{prop:window-FOR},
  so that  $\f^i(\bigsqcup Q^\Phi)$ is computable by Lemma~\ref{lm:fixpoint-FOR}
  as an $\exists$-$\R$ formula of size in
  $|\Phi||\MDP|^{\mathcal O(\Ell+i)}$.

\begin{example}\label{ex:window-fixpoint}
  Let $\MDP$ be the MDP to the right of \figurename~\ref{fig:MDP},
  and $\Phi=\Pr{\FF^2 s_2}\geq\frac{9}{16}$.
  As detailed in Example~\ref{ex:window-FOR}, the set of strategies $(\bigsqcup Q^{\Phi})_{s_0}$
  is described by $x_{s_0,a}x_{s_0as_1,b}\geq \frac{1}{4}$,
  the set of strategies $(\bigsqcup Q^{\Phi})_{s_1}$
  is described by $x_{s_1,a}x_{s_1as_0,b}\geq \frac{1}{4}$,
  and the set $\bigsqcup Q^{\Phi}$ is described by $\top$ on all other states.\footnote{%
  Once again, we omit the constraints that ensure that all variables encode probabilities.}
  By Example~\ref{ex:window-fixpoint-op}, $\f^i(\bigsqcup Q^\Phi)_{s_0}$
  is described by $x_{s_0,a}x_{s_0as_1,b}\geq \frac{1}{4} \wedge x_{s_0as_1, b}\leq c_i$,
  where the constant $c_i$ is defined by $c_0=1$ and $c_{i+1}=1-\frac{1}{4c_i}$.
  Similarly, $\f^i(\bigsqcup Q^\Phi)_{s_1}$
  is described by $x_{s_1,a}x_{s_1as_0,b}\geq \frac{1}{4} \wedge x_{s_1as_0, b}\leq c_i$,
  and $\f^i(\bigsqcup Q^\Phi)$ is $\top$ on all other states.
  The sequence $(c_i)_{i\in\N}$ is a decreasing sequence
  that converges towards $\frac{1}{2}$ (but never reaches it).

  The limit of this sequence is the greatest fixpoint $\Pi_{s_0}^\Phi$, described by
  $x_{s_0,a}x_{s_0as_1,b}\geq \frac{1}{4} \wedge x_{s_0as_1, b}\leq \frac{1}{2}$
  on $s_0$,
  $x_{s_1,a}x_{s_1as_0,b}\geq \frac{1}{4} \wedge x_{s_1as_0, b}\leq \frac{1}{2}$
  on $s_1$, $\top$ everywhere else.
  If we follow the proof of Proposition~\ref{prop:fixed-point}, we can
  recover the only choice on $s_0$ and $s_1$ that ensures $\Forall \GG \Phi$:
  play $a$ and $b$ with probability $\frac{1}{2}$.
\end{example}

  We note that this fixed point computation is not an algorithm:
  as we have seen in Example~\ref{ex:window-fixpoint} the fixed point
  may not be reachable in finitely many steps. In this case, we do not know
  if the limit will be empty or not.
  Nonetheless, this characterisation yields multiple corollary results, that we detail in  the remainder of this section.

\subparagraph*{{\bf Flat, non-strict formulae}}

  If $\Phi$ is flat and non-strict then $\f^i(\bigsqcup Q^\Phi)$ maps
  every state to a compact set.\footnote{%
    Non-strict formulae describe closed sets, and all variables are
    in $[0,1]$ as they encode probabilities.
  }
  The limit of an infinite decreasing sequence of
  non-empty compact sets in $\R^{\mathcal X_s}$ is non-empty.
  Therefore, if the limit of a decreasing sequence of
  compact sets is the empty set, it must be reached after
  finitely many steps. Thus, if $\Pi^\Phi_{s}=\emptyset$,
  then there exists $i\in\N$ so that $\f^i(\bigsqcup Q^\Phi)_s=\emptyset$.

\begin{theorem}\label{thm:coRE-easy}
  The synthesis problem for flat, non-strict \emph{global window} $L$-PCTL formulae
  is in $\co\RE$.
\end{theorem}

  As we will detail in Section~\ref{sec:undec}, the synthesis problem for flat,
  non-strict global window formulae is undecidable ($\co\RE$-hard),
  and therefore $\co\RE$-complete.

\begin{remark}
  From the proof of Proposition~\ref{prop:window-FOR}, it follows that if $\Phi$ is
  non-flat, that is, it contains nested probabilistic operators, then the set of
  window strategies $\stratW$ such that $s \models_\stratW \Phi$ may not be closed
  and hence $\bigsqcup Q^\Phi$ is not necessarily a compact set.
\end{remark}
\begin{remark}
  Note that in Example~\ref{ex:window-fixpoint} we were able to compute by hand
  the limit of the sequence of $\exists$-$\R$ formulae
  describing $\f^i(\bigsqcup Q^\Phi)$, and obtained an $\exists$-$\R$ formula
  for the greatest fixed point $\Pi^\Phi_s$.
  This is not always possible: there exists an MDP $\MDP$ and a flat,
  non-strict global window formula $\Forall \GG \Phi$ so that
  $\Pi^\Phi_s$ cannot be expressed in FO-$\R$.
  Indeed, FO-$\R$ formulae can be seen as finite words over a countable alphabet,
  so that there are countably many of them.
  If by contradiction $\Pi^\Phi_s$ was always expressible in FO-$\R$,
  we could enumerate all FO-$\R$ formulae and check for each of them if it
  describes a fixed-point of $\f$ where $\initState$ is mapped to a non-empty set,
  two properties
  also expressible in FO-$\R$ by using Lemma~\ref{lm:fixpoint-FOR}.
  This would show that the synthesis problem is recursively enumerable,
  therefore in $\RE\cap\co\RE$ \textit{i.e.}~decidable, which is absurd
  as it is $\co\RE$-complete as we will see in Section~\ref{sec:undec}.
\end{remark}

\subparagraph*{{\bf Deterministic strategies}}
In this paragraph, we study the synthesis problem for deterministic strategies.
  First, note that the window strategy defined by a deterministic strategy
  for a given prefix and horizon is also deterministic.
  Conversely, if $\stratW$ is a deterministic window strategy
  then there exists a deterministic strategy in its cylinder.
  Therefore, Lemma~\ref{lm:window-length} carries over,
  and finding a deterministic strategy
  satisfying a window formula reduces to finding a
  deterministic window strategy for it.
  Then, note that for a fixed state $s$, each
  deterministic window strategy can be seen as a boolean assignment over
  the set $\mathcal X_s$ of variables,
  and we have that $|\mathcal X_s| = |\MDP|^{\mathcal O(\Ell)}$.
  Therefore, the set of deterministic window strategies is finite,
  of doubly-exponential size $2^{|\MDP|^{\mathcal O(\Ell)}}$.
  We denote by $W$ the number of deterministic window strategies.
  By guessing a window strategy and verifying it in $\EXP$,
  we get a $\NEXP$ upper bound on the synthesis problem for window formulae.
  By guessing a strategy in an online manner we can lower this complexity,
  and show that the problem is in fact $\PSPACE$-complete.

\begin{proposition}\label{prop:HD-window}
  The synthesis problem for \emph{window} $L$-PCTL formulae is $\PSPACE$-complete
  when restricted to \emph{deterministic} strategies.
\end{proposition}
\begin{proof}%
  We present a non-deterministic algorithm, running in polynomial space,
  that accepts all positive instances of the synthesis problem
  with deterministic strategies.
  We will guess a deterministic window strategy $\stratW$
  and check that $\Phi$ holds on the resulting MC.
  In order to avoid guessing an exponential certificate
  (the entire strategy $\stratW$),
  we will perform a depth-first search (DFS) traversal of the MDP,
  starting from $\initState$ and of horizon $\Ell$,
  where we guess every decision of $\stratW$ in an online manner
  ($\stratW(\rho)$ is guessed when the search path, that is the path from $\initState$ to the current state, is $\rho$ for the first time).
  We will compute along the way information that ultimately lets
  us evaluate if $\Phi$ holds on the root node of the search.
  At any point in the DFS, when the current path traversed from $\initState$ is $\rho$,
  this information represents
  partial evaluations of subformulae of $\Phi$ on states along $\rho$,
  according to the strategy $\stratW$.
  Formally, we equip each state in the DFS with a set of formulae
  to be evaluated. For each path formula $\varphi$ in the set,
  we store the probability
  of satisfying the formula according to paths previously visited by the DFS.
  Once all of the subtree below a state has been seen by the search,
  this value matches the probability $\Pr{\varphi}$,
  and we can then use this value to evaluate the state formulae that needs
  to know $\Pr{\varphi}$ on the current state.
  In order to define the sets of subformulae to evaluate,
  we can use the same induction rules as in
  the proof of Proposition~\ref{prop:window-FOR}, that reduce the evaluation
  of a formula such as $\XX^\ell \Phi_1$ or $\Phi_1 \UU^\ell \Phi_2$
  on a given state to the evaluation of $\Phi_1$, $\Phi_2$,
  $\XX^{\ell-1} \Phi_1$ or $\Phi_1 \UU^{\ell-1} \Phi_2$ on the current
  or the next state.
  Overall, we need to remember a path of length at most $\Ell$,
  a set of subformulae of $\Phi$ on each state in this path,
  and a probability for each such path formula.
  Assuming that the probabilities can be stored in polynomial space, this
  is indeed a $\PSPACE$ algorithm.

  We argue that these probabilities can always be stored in polynomial space.
  Indeed, they correspond to the measure,
  in the MC defined by $\stratW$,
  of a finite union of cylinders
  defined by prefixes of length at most $\Ell$.
  Since $\stratW$ is deterministic,
  these measures are finite %
  sums of real numbers
  obtained as the product of at most $\Ell$ constants appearing as transition probabilities on $\MDP$.
  Using a standard binary representation of rational numbers
  as irreducible fractions, we get that since $\Ell$
  is polynomial in $|\Phi|$ these probabilities are always rational and
  of polynomial size.

We show a reduction from the synthesis problem with a generalized reachability objective in a two-player game. Given an arena with a set $V$ of vertices that are partitioned into vertices belonging to Player~1 and Player~2, given an initial vertex $v_0$, and reachability sets $F_1, \dots , F_k$, the problem asks for a (deterministic) Player~1 strategy that ensures reaching each of the sets against any Player~2 strategies.
The generalized reachability problem is $\PSPACE$-complete~\cite{FH10}.

We construct an MDP $\MDP$ with $Q=V$ set of states and transitions which are the same as the edges of the two-player game arena.
A Player~1 vertex corresponds to a state in the MDP such that for every outgoing edge $(v,v_i)$ from $v$, we have an action $a_i$ labelling the transition $(v,v_i)$ in $\MDP$.
For a Player~2 vertex $v$, all the outgoing edges $(v,v_i)$ correspond to transitions for the same action to vertices $v_i$ with equal probability.
Also a state $v \in Q$ in $\MDP$ is labelled $x_i$ for $i \leq i \leq k$ if and only if the corresponding vertex $v \in F_i$ in the two-player game.
Further, if Player~1 has a wining strategy in the generalized reachability game, then she can visit all the reachability sets within a total of $nk$ steps with a deterministic strategy.

Now consider the property $\Phi$ defined as $\bigwedge_{i=1}^k \Pr{\FF^{nk}x_i}=1$.
There exists a deterministic strategy from $v_0$ in $\MDP$ satisfying $\Phi$ if and only if Player~1 has a winning strategy for the generalized reachability objective.
\end{proof}

  As there are finitely many deterministic window strategies of horizon $\Ell$,
  the fixed point computation always terminates and thus provides
  decidability. We also reduce the problem asking if an alternating Turing machine
  running in polynomial space accepts a given word to deterministic
  strategy synthesis.

\begin{proposition}\label{prop:HD-global-window}
  The synthesis problem for \emph{global window} $L$-PCTL formulae is
  in $2\EXP$ when restricted to \emph{deterministic} strategies.
  Moreover, it is $\EXP$-hard.
\end{proposition}
\begin{proof}%
  For global window formulae, we need to change the set $Q^\Phi$, defined in Section~\ref{sec:fixpoint}, to only
  contain deterministic window strategies.
  It is still a complete lattice, and Proposition~\ref{prop:fixed-point}
  carries over for deterministic strategies.
  Moreover, for every portfolio $\Pi$, there are finitely many
  strictly smaller portfolios, at most $|\MDP| W$.
  As the sequence $(\f^i(\bigsqcup Q^\Phi))_{i\in\N}$ is non-increasing,
  the fixed point is reached in at most $|\MDP| W$ steps.
  Each step can be performed without relying on the theory of the reals,
  by representing the window strategies explicitly as trees of depth $\Ell$.
  Applying the operator $\f$ on a portfolio
  amounts to checking if a tree is a prefix of another.
  Overall, the fixed point computation is doubly-exponential.

  Note that if we rely on computing $\exists$-$\R$ formulae
  for $(\f^i(\bigsqcup Q^\Phi))_{i\in\N}$
  instead of these explicit sets of strategies, the formulae could
  a priori grow to sizes in $|\Phi||\MDP|^{\mathcal O(\Ell+|\MDP|W)}$,
  so that we end up with a triply-exponential upper bound.

\begin{figure}[t]
  \centering
  \scalebox{1}{
  \begin{tikzpicture}[node distance=1.8cm,auto]
  \usetikzlibrary{arrows,positioning,automata,calc}
  \tikzstyle{accepting}=[accepting by arrow]

  \node[player0, initial](qe){$q^a$};
  \node[player0](qem)[right of=qe]{$q\xrightarrow{a,b,L}q_j$};
  \node[player0](qet)[above of=qem]{$q\xrightarrow{a,a,R}q_i$};
  \node[player0](qeb)[below of=qem]{$q\xrightarrow{a,b,R}q_k$};
  \node[player0, accepting](qe')[right of=qem]{};
  \node[player1, dashed]()[right of=qe, minimum size=4.5cm]{};
  \node[]()[above of=qet,node distance=0.7cm]{$q^a$ if $L(q)=\exists$};

  \path (qe) edge[action1] node[right]{\small$m_1$} (qet);
  \path (qe) edge[action2] node[above]{\small$m_2$} (qem);
  \path (qe) edge[action3] node[right]{\small$m_3$} (qeb);
  \path (qet) edge[action0] (qe');
  \path (qem) edge[action0] (qe');
  \path (qeb) edge[action0] (qe');

  \node[player0, initial](qf)[right of=qe']{$q^a$};
  \node[player0](qfm)[right of=qf]{$q\xrightarrow{a,a,R}q_j$};
  \node[player0](qft)[above of=qfm]{$q\xrightarrow{a,b,L}q_i$};
  \node[player0](qfb)[below of=qfm]{$q\xrightarrow{a,b,R}q_k$};
  \node[player0, accepting](qf')[right of=qfm]{};
  \node[player1, dashed]()[right of=qf, minimum size=4.5cm]{};
  \node[]()[above of=qft, node distance=0.7cm]{$q^a$ if $L(q)=\forall$};

  \path (qf) edge[action0] node[right]{\tiny$1/3$} (qft);
  \path (qf) edge[action0] node[above]{\tiny$1/3$} (qfm);
  \path (qf) edge[action0] node[right]{\tiny$1/3$} (qfb);
  \path (qft) edge[action0] (qf');
  \path (qfm) edge[action0] (qf');
  \path (qfb) edge[action0] (qf');

  \node[player0, initial](qa)[right of=qf']{$q^a$};
  \node[player2, dashed]()[right of=qa,node distance=1mm,minimum size=1.3cm]{};
  \node[]()[above of=qa, node distance=1.2cm, xshift=.6cm]{$q^a$ if $L(q) \in \{q_\top, q_\bot\}$};

  \path (qa) edge[action0, loop right] (qa);

  \node[player0](s0)[below of=qe,node distance=5.2cm]{};
  \node[player1, dashed](qm)[right of=s0]{$q_2^a$};
  \node[player1, dashed](qt)[above of=qm,node distance=1cm]{$q_1^b$};
  \node[player1, dashed](qtt)[above of=qt,node distance=1cm]{$q_1^a$};
  \node[](qb)[below of=qm,node distance=1cm]{$\vdots$};
  \node[player2, dashed](qbb)[below of=qb,node distance=1cm]{$q_\top^c$};
  \node[player0](s1)[right of=qm]{$s$};
  \node[player0,initial,initial where=below](init)[below of=s1]{$\initState$};
  \node[player0](s1m)[right of=s1]{$b_1$};
  \node[player0](s1t)[above of=s1m,node distance=1cm]{$a_1^H$};
  \node[player0](s1tt)[above of=s1t,node distance=1cm]{$a_1$};
  \node[](s1b)[below of=s1m,node distance=1cm]{$\vdots$};
  \node[player0](s1bb)[below of=s1b,node distance=1cm]{$c_1^H$};
  \node[player0](s1')[right of=s1m]{};
  \node[](s1'')[right of=s1',node distance=.75cm]{$\cdots$};
  \node[player0](s2)[right of=s1'',node distance=.75cm]{};
  \node[player0](s2m)[right of=s2]{$b_n$};
  \node[player0](s2t)[above of=s2m,node distance=1cm]{$a_n^H$};
  \node[player0](s2tt)[above of=s2t,node distance=1cm]{$a_n$};
  \node[](s2b)[below of=s2m,node distance=1cm]{$\vdots$};
  \node[player0](s2bb)[below of=s2b,node distance=1cm]{$c_n^H$};
  \node[player0](s3)[right of=s2m]{};

  \path (s0) edge[action1,bend left=15] node[above,yshift=1mm]{\small$m_1$} (qtt);
  \path (s0) edge[action2,bend left=5] node[above]{\small$m_2$} (qt);
  \path (s0) edge[action3] node[above]{\small$m_3$} (qm);
  \path (s0) edge[action4,bend right=5] (qb);
  \path (s0) edge[action5,bend right=15] (qbb);
  \path (qtt) edge[action0, bend left=15] (s1);
  \path (qt) edge[action0, bend left=5] (s1);
  \path (qm) edge[action0] (s1);
  \path (qb) edge[action0, bend right=5] (s1);
  \path (init) edge[action0] (s1);
  \path (s1) edge[action1,bend left=15] node[above,yshift=1mm]{\small$m_1$} (s1tt);
  \path (s1) edge[action2,bend left=5] node[above]{\small$m_2$} (s1t);
  \path (s1) edge[action3] node[above]{\small$m_3$} (s1m);
  \path (s1) edge[action4,bend right=5] (s1b);
  \path (s1) edge[action5,bend right=15] (s1bb);
  \path (s1tt) edge[action0, bend left=15] (s1');
  \path (s1t) edge[action0, bend left=5] (s1');
  \path (s1m) edge[action0] (s1');
  \path (s1b) edge[action0, bend right=5] (s1');
  \path (s1bb) edge[action0, bend right=15] (s1');
  \path (s2) edge[action1,bend left=15] node[above,yshift=1mm]{\small$m_1$} (s2tt);
  \path (s2) edge[action2,bend left=5] node[above]{\small$m_2$} (s2t);
  \path (s2) edge[action3] node[above]{\small$m_3$} (s2m);
  \path (s2) edge[action4,bend right=5] (s2b);
  \path (s2) edge[action5,bend right=15] (s2bb);
  \path (s2tt) edge[action0, bend left=15] (s3);
  \path (s2t) edge[action0, bend left=5] (s3);
  \path (s2m) edge[action0] (s3);
  \path (s2b) edge[action0, bend right=5] (s3);
  \path (s2bb) edge[action0, bend right=15] (s3);

  \draw [->] (s3) to [out=90,in=0] ($(s2tt)+(0,.7)$) to [out=180,in=0] ($(qtt)+(0,.7)$) to [out=180, in=90 ] (s0);

  \end{tikzpicture}
  }
  \caption{The MDP used in the reduction from $\APSPACE$ to deterministic synthesis for global window formulae.
  The colored transitions $m_1$, $m_2$, \dots represent different actions,
  black transitions are available for every action,
  and the probability of a transition is $1$ if unspecified.
  The letters $a$, $b$, $c$ enumerate the tape alphabet $\Sigma$,
  while $q_1$, $q_2$, \dots enumerate the states in $Q$,
  with $q_\top$ the accepting state.
  The only randomized transition is in the universal gadget $q^a$,
  and uses a uniform distribution over reachable states.
}
  \label{fig:MDP-EXP}
  \end{figure}
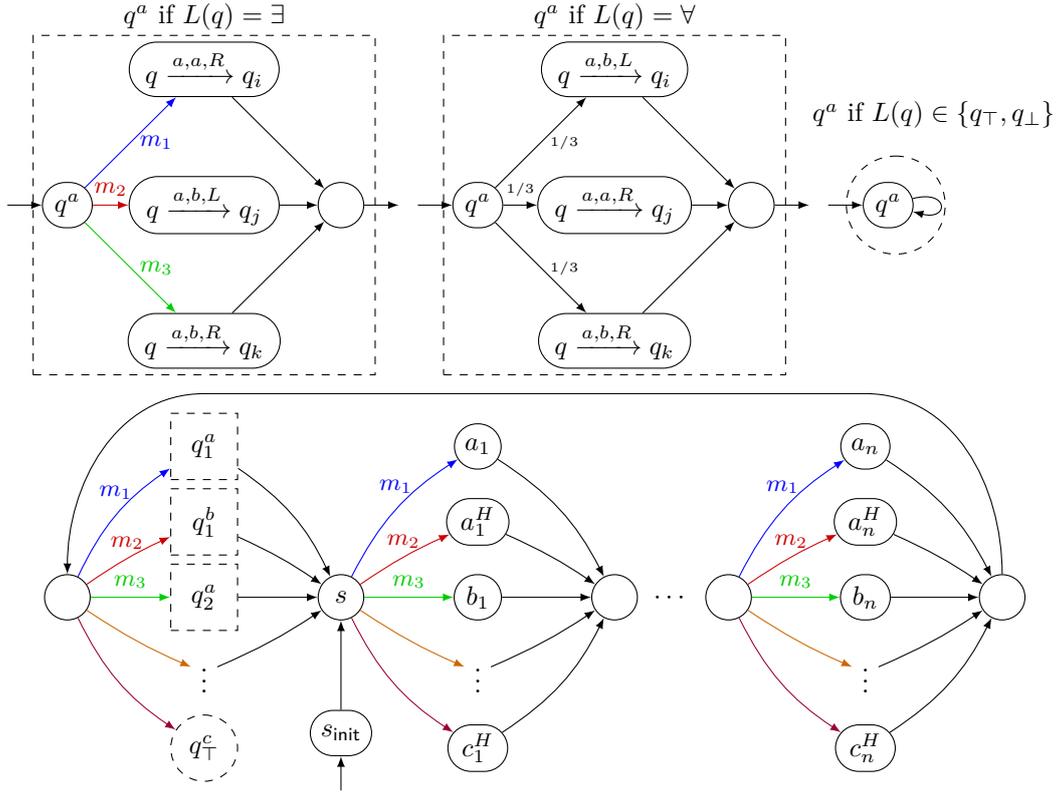

  For the $\EXP$-hardness, we use $\APSPACE=\EXP$ and present a polynomial reduction from
  alternating polynomial-space Turing machines.
  We consider a Turing machine of states $Q$ and tape alphabet $\Sigma$,
  so that each state $q$ is equipped with a label $L(q)\in\{\forall,\exists\}$,
  except for the accepting and rejecting states $q_\top$ and $q_\bot$, where $L(q)=q$.
  Let $w\in\Sigma^*$ be an input word, and let $n\in\N$ be
  a bound on the length of the tape used when running $w$ on the machine.
  Since we considered a polynomial space machine, $n$ is polynomial.
  W.l.o.g., we assume that for every input, the Turing machine we consider above halts, and the input is accepted if and only if it halts in $q_\top$.

  Let $\MDP$ be the MDP from \figurename~\ref{fig:MDP-EXP},
  where each named state is assigned an identical label.
  We describe a window formula $\Phi$ that ensures that controller only makes
  choices that faithfully represent an execution of the alternating
  Turing machine.
  Let $N=2n+5$ be the number of steps needed to follow a cycle
  from $s$ to $s$ in $\MDP$.
  If $l$ is a label in $\MDP$, %
  we shorten the $L$-PCTL formula $\Pr{\FF^N l}=1$
  as $\FF^N_{1} l$. It means that every path of length $N$
  must reach $l$.
    Intuitively, a run of the Turing machine can be described as a path
  in this MDP, where one full loop around $s$ describes a configuration
  of the Turing machine: visiting $a_i$ means that cell number $i$ contains $a$,
  visiting $a_i^H$ means that additionally the reading head is in position $i$,
  and entering the gadget $q^a$ means that we are in state $q$ and will read an $a$.
  Either controller or the environment gets to pick the next transition, and then we go back to $s$.

  The formula $\Phi$ is obtained as the conjunction of the following constraints:

  In order to ensure that the tape is initialized appropriately, we ask for every letter $a \in \Sigma$ in position $i>1$
  in the input word $w$ that $\initState \Rightarrow \FF^N_{1} a_i$. If the first letter in $w$ is $a$ and the initial state of the Turing machine is $q_1$, we also ask $\initState \Rightarrow \FF^N_{1} a_1^H \wedge \FF^N_{1} q_1^a$.

  In order to
  simulate the transitions on the tape cells correctly,
  we ask for every $1\leq i< n$, $a,c\in\Sigma$
  and transition $q\xrightarrow{a,b,L}q'$
  that
  \[(s\wedge \FF^N_{1} c_i \wedge \FF^N_{1} a_{i+1}^H) \Rightarrow \PrBig{\XX^{N-2} ((q\xrightarrow{a,b,L}q')\Rightarrow \FF^N_{1} c_{i}^H \wedge \FF^N_{1} b_{i+1})}=1\,.\]
  We similarly ask for every $1\leq i< n$, $a,c\in\Sigma$
  and transition $q\xrightarrow{a,b,R}q'$
  that
  \[(s\wedge \FF^N_{1} a_i^H \wedge \FF^N_{1} c_{i+1}) \Rightarrow \PrBig{\XX^{N-2} ((q\xrightarrow{a,b,R}q')\Rightarrow \FF^N_{1} b_{i} \wedge \FF^N_{1} c_{i+1}^H)}=1\,.\]

  The other tape cells should be left untouched,
  so that for every $1\leq i< n$, $a,b\in\Sigma$
  and transition $q\xrightarrow{c,d,L}q'$, we ask that
  \[(s\wedge \FF^N_{1} a_i \wedge \FF^N_{1} b_{i+1}) \Rightarrow \PrBig{\XX^{N-2} ((q\xrightarrow{c,d,L}q')\Rightarrow \FF^N_{1} a_{i})}=1\,.\]
  Similarly for every transition transition $q\xrightarrow{c,d,R}q'$, we ask that
  \[(s\wedge \FF^N_{1} a_i \wedge \FF^N_{1} b_{i+1}) \Rightarrow \PrBig{\XX^{N-2} ((q\xrightarrow{c,d,R}q')\Rightarrow \FF^N_{1} b_{i+1})}=1\,.\]

  Finally, in order to update the state, we ask for every transition $q\xrightarrow{a,b,D}q'$ with $D\in \{L,R\}$ and every $1\leq i\leq n$, $c\in\Sigma$ that
  \( (q\xrightarrow{a,b,D}q') \wedge \FF^N_{1} c_i^H \Rightarrow \FF^N_{1} q'^c\,.\)

  Then, we let $q_\top$ and $q_\bot$ be the labels that hold on
  all states $q_\top^a$ and $q_\bot^a$, respectively, for all $a\in\Sigma$.
  We consider the global window formula
  $\Forall\GG[(\Phi\vee q_\top) \wedge \neg q_\bot]$,
  and show that there is a winning strategy for this formula in $\MDP$
  if and only if the alternating Turing machine
  accepts the input word $w$.
  Indeed, an alternating Turing machine equipped
  with an initial word can be seen as
  a turn-based two-player zero-sum reachability game
  played on the execution tree of the machine, where we ask
  if there exists a strategy for player $\exists$
  that ensures against every strategy of $\forall$
  the state $q_\top$ is reached.
\end{proof}

\subparagraph*{{\bf Memoryless strategies}}
  We study the synthesis of memoryless strategies.
  The window strategy defined by a memoryless strategy
  for a given prefix and a horizon is also memoryless.
  Conversely, a memoryless window strategy has a memoryless strategy
  in its cylinders.
  By Lemma~\ref{lm:window-length}, finding a memoryless strategy
  satisfying a window formula reduces to finding a
  memoryless window strategy for it.
  Let $s$ be a state.
  As usual, a window strategy $\stratW$ for state $s$ can be seen as
  an assignment in $[0,1]$ for variables $\mathcal X_s$.
  However, the memoryless property asks that $\stratW(\rho)=\stratW(\rho')$
  for all $\rho,\rho'$ that share the same last state $s'$, or equivalently
  $x_{\rho,a}=x_{\last(\rho),a}$ for all $\rho$.
  Thus, we can replace every instance of $x_{\rho,a}$
  by $x_{\last(\rho),a}$ in the $\exists$-$\R$ formula
  of Proposition~\ref{prop:window-FOR}, so that
  the set of free variables used to represent
  a memoryless window strategy for $s$ is $\mathcal X_s=
  \{x_{s',a}\mid \exists\rho\in\FPaths_\MDP^{<\Ell}(s), s'=\last(\rho)\}$.
  Similarly, the variables $y_{\rho,\Phi}$ and $z_{\rho,\varphi}$
  can be replaced by by $y_{\last(\rho),\Phi}$ and $z_{\last(\rho),\varphi}$
  respectively, as the satisfaction of a state formula,
  or the probability of satisfying a path formula,
  only depend on the current state.
  The formula is now of polynomial size,
  so that we obtain as a corollary:
\begin{proposition}\label{prop:memoryless-window-synthesis}
  The synthesis problem for \emph{window} $L$-PCTL formulae is in $\PSPACE$
  when restricted to \emph{memoryless} strategies.
\end{proposition}

  Further, following a reduction in \cite{10.1109/LICS.2006.48}, it can be shown that the \MR synthesis problem for window $L$-PCTL objectives is at least as hard as the {\sc Square-Root-Sum} problem which is known to be in \PSPACE, but whose exact complexity is a longstanding open problem.

  We now study the memoryless synthesis problem for global window formulae.
  For each state $s$, let $R^s$ denote the $\exists$-$\R$ formula encoding
  the window formula $\Phi$ for state $s$,
  as per Proposition~\ref{prop:memoryless-window-synthesis}.
  The free variables are the variables
  in $\mathcal X_s\subseteq \mathcal X=\{x_{s',a} \mid s'\in S, a\in A\}$.
  A memoryless strategy $\sigma$ can be seen as a point in $\R^{\mathcal X}$,
  so that $\sigma(s,a)$ is
  assigned to $x_{s,a}$.
  For all states $s$ and $s'$, we define a variable $r_{s,s'}\in\{0,1\}$ quantified existentially,
  and construct a
  formula ensuring that if $r_{s,s'}=0$ then
  $s'$ is not reachable from $s$
  under the strategy $\sigma$.
  This formula states $r_{s,s}=1$ for all states $s$, and
  asks that the variables $r$ are a solution
  to the system of equations asking, for all $s,s',s''$ and $a$, that if $r_{s,s'}=1$ and $x_{s',a}\proba(s',a,s'')>0$ then $r_{s,s''}=1$.
  Therefore, the set of states $s'$ so that $r_{s,s'}=1$ is an
  over-approximation.\footnote{For example, the formula
  is satisfied if $r_{s,s'}$ is 1 for all $s,s'$, which represents an over-approximation of the set of states reachable from $s$ where every state is reachable.}

  Then, the formula asking that there exists a value for each variable $r$ so that $R^{s'}$ holds
  whenever $r_{s,s'}=1$
  represents the
  memoryless strategies that satisfy $\Phi$ on an over-approximation of the states reachable from $s$,
  which is equivalent to satisfying
  $\Forall \GG \Phi$ when starting from
  state $s$.
  Note that since the variables $r_{s,s'}$ are existentially quantified, and $\Phi$ is only required to be satisfied on states reachable from $s$, then there always exists a valuation for these $r$ variables that sets $r_{s,s'}$ to $1$ if and only if $s'$ is reachable from $s$.
  It follows that:

\begin{proposition}\label{prop:MR-global-window}
  The synthesis problem for \emph{global window} $L$-PCTL formulae is
  in $\PSPACE$ when restricted to \emph{memoryless} strategies.
\end{proposition}

\subparagraph*{{\bf PCTL satisfiability}} We now consider the satisfiability problem, that asks,
  given a formula $\Phi$, if there a exists an MC $M$
  so that $M\models\Phi$.
  This is a longstanding open problem for PCTL formulae.
  One can also consider variants of the problem, that either restrict $\Phi$
  to a sublogic of PCTL or limit $M$ to MCs that belong to
  a particular set, such as finite MCs or MCs where
  all probabilities are rational numbers.
  The decidability of these variants is also open and, as noted in~\cite{DBLP:conf/lics/BrazdilFKK08},
  some PCTL formulae are only satisfiable by infinite MCs.
  In particular, we say that an MC $M$ has granularity bounded by $N\in\N$ if
  every probability in the transition function $\proba$ is equal to a rational $\frac{a}{b}$
  with $b\leq N$.
  The \emph{bounded granularity satisfiability problem} asks, given $\Phi$ and $N$,
  if there exists an MC of granularity bounded by $N$ that satisfies $\Phi$.

  The bounded granularity satisfiability problem for global window $L$-PCTL formulae
  can be reduced to the {\sf HD} strategy synthesis problem for
  global window $L$-PCTL formulae.
  Therefore, we obtain the following result as a corollary of Proposition~\ref{prop:HD-global-window}:
\begin{theorem}\label{thm:satisfiability}
  The bounded granularity satisfiability problem for global window $L$-PCTL formulae
  is decidable in complexity doubly-exponential in both $|\Phi|$ and $N$. Moreover, finite MCs are sufficient, in the sense
  that for every formula $\Forall \GG \Phi$ that admits a model $M$
  of granularity bounded by $N$, there exists a finite MC $M'$
  of granularity bounded by $N$ so that $M'\models \Forall \GG \Phi$.
\end{theorem}
\begin{proof}[Proof sketch]
  Given a formula with atomic propositions \AP, and a granularity bound $N$,
  we intuitively consider an MC of states $2^\AP$
  with an action for every distribution over $2^\AP$ whose granularity is bounded by $N$,
  so that this action describes the next states and their probabilities.
  Then, every MC of granularity bounded by $N$ can be seen as a deterministic
  strategy in this MDP, so that strategy synthesis and MC satisfiability
  are equivalent. We can then apply Proposition~\ref{prop:HD-global-window}. %
  Moreover, finite MCs are sufficient as finite-memory strategies
  are sufficient for global window PCTL when restricted to deterministic strategies.
\end{proof}

\section{Undecidability}\label{sec:undec}

  In Section~\ref{sec:fixpoint}, we have shown that the synthesis problem for flat, non-strict global window $L$-PCTL formulae is in $\co\RE$. In this section, we argue that it is $\co\RE$-hard %
  and that it becomes $\Sigma_1^1$-hard when relaxing the hypothesis that the formulae considered are non-strict.

  When considering flat non-strict formulae, we proceed via a reduction from the non-halting problem of a two-counter Minsky machine. A two-counter Minsky machine consists of a list of instructions $l_1:\mathsf{ins}_1,\ldots,l_n:\mathsf{ins}_n$ and two counters $c_2$ and $c_3$ (the indices $2$ and $3$ are chosen to ease the notations) where, for all $i \leq n$, we have $\mathsf{ins}_i$ an instruction in one the following types, for $j \in \{ 2,3 \}$ and $1 \leq k,m \leq n$: $\mathsf{Inc}_j(k)$: $c_j := c_j + 1$; goto $k$; $\mathsf{Branch}_j(k,m)$: if $c_j = 0$ then goto $k$; else $c_j := c_j - 1$, goto m; \textsf{H}: halt. The semantics of these instructions is straightforward. %
  The non-halting problem for Minsky machine, denoted \textsf{MinskyNotStop}, is to decide, given a machine $\mathsf{Msk}$, if its execution is infinite. This problem is undecidable, as stated in the theorem below.

\begin{theorem}[\cite{Minsky67}]
  \label{thm:minsky_co_re_hard}
  \textsf{MinskyNotStop} is \co\RE-complete.
\end{theorem}

\begin{figure}
  \begin{minipage}{0.2\linewidth}
    \centering
    \includegraphics[scale=0.8]{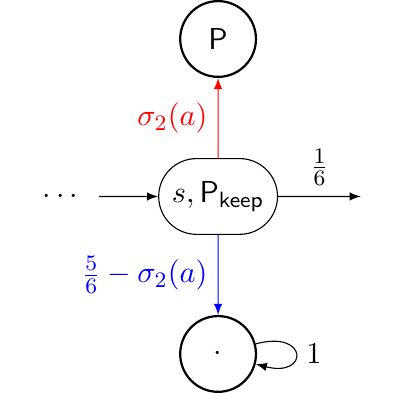}
    \caption{The end of a gadget.}
    \label{ExitGadget}
  \end{minipage}
  \begin{minipage}{0.5\linewidth}
    \centering
    \includegraphics[scale=0.8]{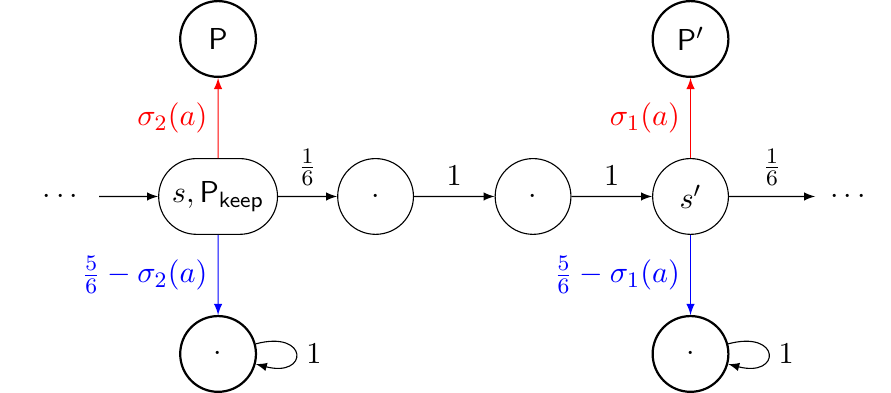}
    \caption{The end of a gadget on the left, and the beginning of another one on the right.}
    \label{BeginEndGadget}
  \end{minipage}
  \begin{minipage}{0.2\linewidth}
    \centering
    \includegraphics[scale=0.8]{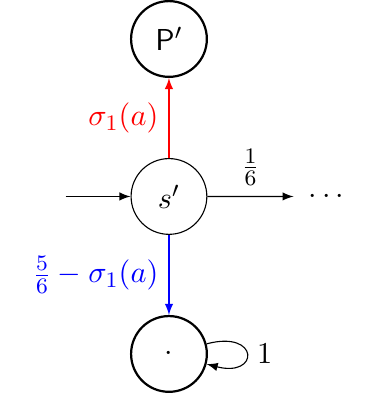}
    \caption{The beginning of a gadget.}
    \label{BeginGadget}
  \end{minipage}
\end{figure}

  Given $\mathsf{Msk} = l_1:\mathsf{ins}_1,\ldots,l_n:\mathsf{ins}_n$ on two counters $c_2$ and $c_3$, we build an MDP $\MDP$ and an $L$-PCTL formula $\Phi$ such that there exists a strategy $\sigma$ for $\MDP$ s.t. $\MDP[\sigma] \models \Phi$ if and only if $\mathsf{Msk} \in \mathsf{MinskyNotStop}$.
  The crucial point of the reduction is to encode the values of the counters that may take unbounded values. It is done in $\MDP$ by encoding these values in the probability (chosen by the strategy $\sigma$) to see a given predicate in the next few steps. More specifically, in the situation where the counters are such that $\{ c_2 \mapsto x_2; c_3 \mapsto x_3 \}$, we consider the probability $p(x_2,x_3) = \frac{5}{6} \times \frac{1}{2^{x_2}} \times \frac{1}{3^{x_3}}$. We then associate to each different instruction a gadget, i.e. an MDP, and a formula encoding the update of probability $p(x_2,x_3)$ according to how the counters are changed by the corresponding instruction. Inside a gadget, one can find predicates of the shape $(\mathsf{P}_\cdot)$. They are used to define the formulae specifying the expected behavior of the strategy. Furthermore, there is also an entering and an exiting probability which correspond to the encoding of the counters respectively before and after the effect of the instructions. We define below formally %
  the notion of well-placed gadgets.
\begin{definition}[Gadgets]
  A gadget $\mathsf{Gd}$ is an MDP with an \emph{entering probability} and an \emph{exiting probability}. Consider Figure~\ref{ExitGadget} that represents how every gadget $\mathsf{Gd}$ ends. The exiting probability $p^\mathsf{ex}_\mathsf{Gd}$ %
  is the probability $\sigma_2(a)$ %
  to visit the state on the top. %
  It is equal to $p^\mathsf{ex}_\mathsf{Gd}=\mathbb{P}_s(\FF^{1} \mathsf{P})$,
  \emph{i.e.}~the probability that $\FF^{1} \mathsf{P}$ holds on state $s$.
    Consider Figure~\ref{BeginGadget}. All gadgets %
    begin as in this figure: a state $s'$ %
  with a successor satisfying the predicate $\mathsf{P}'$. %
  The \emph{entering probability} $p_{\mathsf{Gd}}^{\mathsf{en}}$ %
  is the probability $\sigma_1(a)$ to see $\mathsf{P}'$, that is: $p_{\mathsf{Gd}}^{\mathsf{en}} = \mathbb{P}_{s'}(\FF^{1} \mathsf{P}')$. %
  A gadget is \emph{well-placed} if, as for the gadget on the right of Figure~\ref{BeginEndGadget}, it is preceded by two dummy states, themselves immediately preceded by a gadget.
\end{definition}

  Before looking at how specific instructions are encoded in the counters and the formula, we have to ensure that the exiting probability of a gadget is equal to the entering probability of the following well-placed gadget. This is done with the formula: $\Phi_\mathsf{keep} := \mathsf{P}_\mathsf{keep} \Rightarrow (\mathbb{P}(\FF^{1} \mathsf{P}) = 6 \cdot  \mathbb{P}(\FF^{4} \mathsf{P}'))$.
  These definitions ensure the following proposition:
\begin{proposition}[Entering probability of a well-placed gadget]
  Assume that a well-placed gadget $\mathsf{Gd}'$ follows a gadget $\mathsf{Gd}$.%
  Then, for a strategy $\sigma$ s.t. the formula $\Forall \GG \Phi_\mathsf{keep}$ is satisfied, the exiting probability of gadget $\mathsf{Gd}$ is equal to the entering probability of gadget $\mathsf{Gd}'$: $p_{\mathsf{Gd}}^{\mathsf{ex}} = p_{\mathsf{Gd}'}^{\mathsf{en}}$.
  \label{prop:gadget_well_placed}
\end{proposition}

\begin{figure}
  \begin{minipage}{0.4\linewidth}
    \centering
    \includegraphics[scale=0.8]{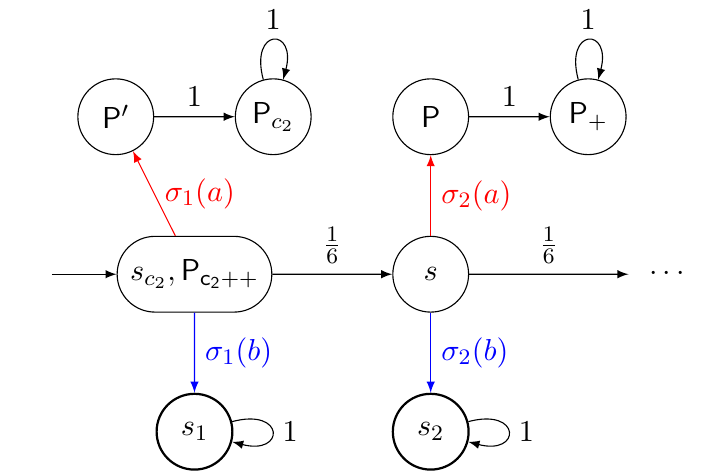}
    \caption{The gadget $\mathsf{Gd}_{c_2++}$ for the operation $c_2 := c_2 + 1$.}
    \label{GadgetIncC2}
  \end{minipage}
  \begin{minipage}{0.45\linewidth}
    \centering
    \includegraphics[scale=0.8]{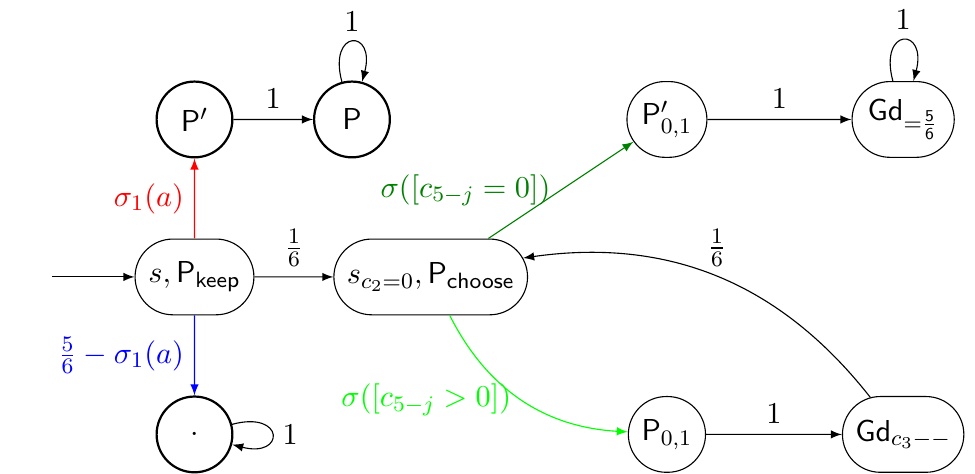}
    \caption{The gadget $\mathsf{Gd}_{c_2 = 0}$ for testing that the counter $c_2 = 0$.}
    \label{GadgetTestEq0}
  \end{minipage}
\end{figure}

  Due to lack of space, we only exhibit the gadgets encoding the increment of a counter and for testing if a counter value is 0.
  Consider the %
  increment of counter $c_2$. %
  By definition of $p(x_2,x_3)$, incrementing that counter
  is simulated by multiplying the probability by $\frac{1}{2}$.%
  We define the gadget $\mathsf{Gd}_{c_2 ++}$ and the formula $\Phi_{c_2 ++}$ ensuring that the probability is indeed multiplied by $\frac{1}{2}$. %
  The gadget $\mathsf{Gd}_{c_2++}$ is depicted in Figure~\ref{GadgetIncC2}.%
  In addition, we define the $L$-PCTL formula $\Phi_{c_2++}$ such that %
  $\Phi_{c_2++} := \mathsf{P}_{c_2++} \Rightarrow (\mathbb{P}(\FF^2 \mathsf{P}_{c_2}) = 6 \cdot 2 \cdot \mathbb{P}(\FF^3 \mathsf{P}_{+}))$. The interest of these definitions lies in the proposition below.
\begin{proposition}[Incrementing Gadget Specification]
  If the entering probability $p_{\mathsf{Gd}_{c_2++}}^{\mathsf{en}}$ of the gadget $\mathsf{Gd}_{c_2++}$ is equal to $p(x_2,x_3)$ with $x_2,x_3 \in \mathbb{N}$, then whenever the formula $\Forall \GG \Phi_{c_2++}$ is satisfied, the exiting probability $p_{\mathsf{Gd}_{c_2++}}^{\mathsf{ex}}$ of this gadget %
  is equal to $p_{\mathsf{Gd}_{c_2++}}^{\mathsf{ex}} = p(x_2+1,x_3)$.
  \label{prop:incrementing_gadget}
\end{proposition}

  We now consider the gadget that tests if a counter value is 0, let us exemplify it with counter $c_2$ in gadget $\mathsf{Gd}_{c_2 = 0}$ depicted in Figure~\ref{GadgetTestEq0}. %
  The gadget $\mathsf{Gd}_{=5/6}$ used on the right tests that both counters have value 0 (i.e. the entering probability is equal to $5/6$).%
  Then, the idea is as follows: as long as counter $c_3$ has a positive value, the strategy $\sigma$ has to take the bottom branch to decrement it and once this counter has reached 0, it can take the top branch to check that the probability is indeed equal to $5/6$. Note that one cannot decrement counter $c_2$ in this gadget, hence if its value is positive, there is no way to pass the test of the comparison to $5/6$. To ensure that the choice at state $s_{c_2 = 0}$ is deterministic, we consider the formula $\Phi_{0,1} :=
  \mathsf{P}_{\mathsf{choose}} \Rightarrow (\mathbb{P}(\XX \mathsf{P}_{0,1}) = 1 \lor \mathbb{P}(\XX \mathsf{P}'_{0,1}) = 1)$. We have the proposition below:
\begin{proposition}[Testing Gadget Specification]
  Assume that the entering probability $p_{\mathsf{Gd}_{c_2=0}}^{\mathsf{en}}$ of the gadget $\mathsf{Gd}_{c_2 = 0}$ is equal to $p(x_2,x_3)$ for some $x_2,x_3 \in \mathbb{N}$. Then, there is a strategy $\sigma$ such that the formula $\Forall \GG [\Phi_{\mathsf{keep}} \wedge \Phi_{0,1} \wedge \Phi_{c_3--} \wedge \Phi_{= 5/6}]$ is satisfied if and only if $x_2 = 0$.
  \label{prop:testing_gadget}
\end{proposition}

  We can similarly test that a counter is different from 0. A gadget corresponding to a branching instruction can then be defined by using these gadgets.%
  Finally, a gadget corresponding to the $\mathsf{Halt}$ instruction only consists of a gadget where no strategy $\sigma$ can satisfy the formula considered. Overall, we can combine all these gadgets to encode all the instructions of the Minsky machine $\mathsf{Msk}$.%
  We obtain the theorem below.
\begin{theorem}
  \label{thm:gwpctl_co_re_hard}
  The synthesis problem for flat, non-strict global window $L$-PCTL formulae
  is \co\RE-hard.
\end{theorem}

  When considering arbitrary flat formulae without the non-strict constraint, we can adapt the proof to reduce from the problem asking if there is an execution of a Minsky Machine that visits infinitely often the first instruction~\cite{AH94}, so that the strategy synthesis problem becomes highly undecidable.

\begin{theorem}
  \label{thm:gwpctl_strict_sigma_one_one_hard}
  The synthesis problem for flat global window $L$-PCTL formulae
  is $\Sigma_1^1$-hard.
\end{theorem}
  The construction here is similar to the case with the non-strict constraint, except that whenever the first instruction is seen, a choice is given to the strategy which can set in how many number of steps $n$ the first instruction will be seen again (note that this number may be arbitrarily large). This choice is encoded by resetting a new counter $c_5$ to value $n$, which is then decremented each time the first instruction is not seen, and a problem arises if this counter ever reaches 0. In terms of probability, the value $p(x_2,x_3)$ is initially multiplied by $\frac{1}{5^n}$ and then multiplied by 5 each time the first instruction is not seen. Hence, the probabilities chose by $\sigma$ may be arbitrarily close to 0, but cannot be equal to 0. This is where we need the non-strict
  comparison with 0.

\clearpage

\begin{center}
    {\Large \bf Appendix}
\end{center}
\appendix

\section{Proofs of Section~\ref{sec:fixpoint}}\label{app:fixpoint}

\begin{lemma}\label{lm:semantics-cylinders}
  For every path formula $\varphi$, MC $\MC$
  and state $s$, the set of paths $\sem{\varphi}_\MC^s$ is measurable.
  Moreover, if the horizon label of $\varphi$ is $\ell\in\N$,
  then $\sem{\varphi}_\MC^s$ is a union of cylinders
  generated from a set $\Cyl\subseteq\FPaths^{\ell}_\MC(s)$ of prefixes,
  so that $\mu_\MC(\sem{\varphi}_\MC^s)=
  \sum_{\rho\in\Cyl} \mu_\MC(\IPaths_\MC(\rho))$.%
\end{lemma}
\begin{proof}%
  We show that
\begin{enumerate}
  \item\label{it:X-cyl} $\sem{\XX^{\ell} \Phi}_\MC^s$ is a finite union
  of disjoint cylinders generated from
  the paths $\rho\in\FPaths^{\ell}_\MC(s)$ so that
  $\rho[\ell]\models_\MC \Phi$.
  \item\label{it:Ul-cyl} %
  $\sem{\Phi_1 \UU^{\ell} \Phi_2}_\MC^s$ is
  a finite union of disjoint cylinders generated from
  the paths $\rho\in\FPaths^{\ell}_\MC(s)$ so that
  $\bigvee_{0\leq j\leq\ell} \left(\rho[j]\models_\MC \Phi_2 \wedge
  \bigwedge_{0\leq i<j} \rho[i] \models \Phi_1\right).$
  \item\label{it:U-cyl} $\sem{\Phi_1 \UU \Phi_2}_\MC^s$ is
  a countable union of cylinders generated from
  the paths $\rho\in\FPaths_\MC(s)$ so that
  $\last(\rho)\models_\MC \Phi_2 \wedge
  \bigwedge_{0\leq i<|\rho|} \rho[i] \models \Phi_1$.
  \item\label{it:Wl-cyl} %
  $\sem{\Phi_1 \WW^{\ell} \Phi_2}_\MC^s$ is
  a finite union of disjoint cylinders generated from
  the paths $\rho\in\FPaths^{\ell}_\MC(s)$
  used for $\Phi_1 \UU^{\ell} \Phi_2$,
  in addition to those so that
  $\bigwedge_{0\leq i\leq \ell} \rho[i] \models \Phi_1$.
  \item\label{it:W-cyl} $\sem{\Phi_1 \WW \Phi_2}_\MC^s$
  is the complement of
  the countable union of cylinders obtained for
  $(\neg\Phi_1 \wedge \neg \Phi_2) \UU \neg\Phi_1$.
\end{enumerate}
  Points~\ref{it:X-cyl} and~\ref{it:Ul-cyl} are directly obtained
  from Definition~\ref{def:semantics}.
  Point~\ref{it:U-cyl} differs
  from point~\ref{it:Ul-cyl} as the semantics definition constrains
  infinite paths in this case, and therefore does not define
  cylinders directly. Instead, we note that every path satisfying
  $\Phi_1 \UU \Phi_2$ admits a finite prefix $\rho$ ensuring
  $\last(\rho) \in \sem{\Phi_2}_\MC \wedge \forall i<|\rho|,
  \rho[i] \in \sem{\Phi_1}_\MC$.
  Points~\ref{it:Wl-cyl} and~\ref{it:W-cyl} use the equivalences
  $\Phi_1 \WW^{\ell} \Phi_2
  \equiv (\Phi_1 \UU^{\ell} \Phi_2) \vee \GG^\ell \Phi_1
  \equiv \neg\left((\neg\Phi_1 \wedge \neg \Phi_2)
  \UU^{\ell} \neg\Phi_1\right)$ to reduce $\cdot \WW^{\ell} \cdot$
  to $\cdot \UU^{\ell} \cdot$
\end{proof}

\begin{proof}[Proof of Lemma~\ref{lm:window-length}]
  Let $\strat': \FPaths_{\MDP} \to \Dist(A)$ be a strategy in the cylinder of
  the window strategy $\stratW_s$.
  Let us show by structural induction on $\Phi,\varphi$ that
  $s\models_{\strat}\Phi \Leftrightarrow s\models_{\strat'}\Phi$ and
  $\mu_{\strat}(\sem{\varphi}_{\MDPstrat}^s) =
  \mu_{\strat'}(\sem{\varphi}_{\MDPtoMC{\MDP}{\strat'}}^s)$.

  We detail the $\Phi_1 \UU^{\ell} \Phi_2$ case, all other cases
  are either trivial or similar.
  Assume that the inductive hypothesis holds for $\Phi_1$ and $\Phi_2$,
  and that $\MDPstrat$ and $\MDPtoMC{\MDP}{\strat'}$
  coincide for the first $\ell + \max(\Ell_1,\Ell_2)$ steps starting from $s$,
  with $\Ell_1$ and $\Ell_2$ the window lengths of $\Phi_1$ and $\Phi_2$.
  Then, consider $\rho\in\IPaths_{\MDP}(s)$ so that
  $\rho[:\ell]\in \FPaths_{\MDPstrat}^{\ell}(s)
  = \FPaths_{\MDPtoMC{\MDP}{\strat'}}^{\ell}(s)$.
  For all $i\leq \ell$, we have that
  $\MDPstrat$ and $\MDPtoMC{\MDP}{\strat'}$
  have the same unfolding of depth $\Ell_1$ starting from $\rho[i]$, so that
  $\rho[i]\in\sem{\Phi_1}_{\MDPstrat}\Leftrightarrow
  \rho[i]\in\sem{\Phi_1}_{\MDPtoMC{\MDP}{\strat'}}$
  by induction hypothesis, and similarly for $\Phi_2$.
  Note that by Lemma~\ref{lm:semantics-cylinders} both
  $\sem{\Phi_1 \UU^{\ell} \Phi_2}_{\MDPstrat}^s$ and $\sem{\Phi_1 \UU^{\ell}
  \Phi_2}_{\MDPtoMC{\MDP}{\strat'}}^s$ are finite unions
  of cylinders generated by paths of length $\ell$. Since $\MDPstrat$
  and $\MDPtoMC{\MDP}{\strat'}$ coincide on the first $\ell$ steps,
  these two sets of paths have the same measure in their respective
  MCs, letting us conclude.
\end{proof}

\begin{proof}[Proof of Proposition~\ref{prop:window-FOR}]
  We say that an $L$-PCTL formula $\Phi_1$ is \emph{contained in} a state formula
  $\Phi_2$ (resp. a path formula $\varphi$) if $\Phi_1$ is a subformula of $\Phi_2$
  (resp. of $\varphi$).
  In order to express $s\models_{\stratW} \Phi$ in the theory of the reals,
  we distinguish three sets of subformulae of $\Phi$,
  based on their position in~$\Phi$.
  On the one hand, the set $\mathcal P$ holds the path formulae
  that are in $\Phi$. On the other hand,
  the state formulae in $\Phi$ are partitioned between two sets
  $\mathcal S_0$ and $\mathcal S_1$, so that state formulae contained
  in some path formula in $\mathcal P$ are sent to $\mathcal S_1$,
  while state formulae that are not subformulae of any path formula
  are sent to $\mathcal S_0$.
  For example, if $\Phi=\Pr{\XX \Pr{\XX^2 p_1}\geq \frac{1}{2}}\leq \frac{1}{2}
  \vee p_2$
  then $\mathcal P=\{\XX \Pr{\XX^2 p_1}\geq \frac{1}{2},\XX^2 p_1\}$,
  $\mathcal S_1=\{\Pr{\XX^2 p_1}\geq \frac{1}{2}, p_1\}$, and
  $\mathcal S_0=\{\Phi,p_2\}$.
  We note that the formula $\Phi$ is always contained in $S_0$.
  Moreover, if a path formula $\XX^{\ell} \Phi_1$
  (resp.~$\Phi_1 \UU^{\ell} \Phi_2$, $\Phi_1 \WW^{\ell} \Phi_2$)
  is in $\mathcal P$, we also add to $\mathcal P$ every variant $\XX^{\ell'} \Phi_1$
  (resp.~$\Phi_1 \UU^{\ell'} \Phi_2$, $\Phi_1 \WW^{\ell'} \Phi_2$)
  with a smaller horizon $0\leq\ell'<\ell$.\footnote{%
    $\XX^{0} \Phi_1$, $\Phi_1 \UU^{0} \Phi_2$ and $\Phi_1 \WW^{0} \Phi_2$ will be used
    for induction purposes, they can be thought of as equivalent to
    $\Phi_1$, $\Phi_2$ and $\Phi_1 \vee \Phi_2$, respectively.}
  Then, we define the following sets of real variables:
\begin{itemize}
  \item A window strategy $\stratW$ for state $s$ is encoded
  using $\mathcal X_s=\{x_{\rho,a}\mid
  \rho\in\FPaths_\MDP^{<\Ell}(s),a\in A\}$,
  so that $x_{\rho,a}$ contains the value of $\stratW(\rho,a)$.
  \item The satisfaction of PCTL state formulae in $\mathcal S_1$
  is encoded using $\mathcal Y_s=\{y_{\rho,\Phi'}\mid
  \rho\in\FPaths_\MDP^{<\Ell}(s), \Phi'\in\mathcal S_1\}$,
  so that $y_{\rho,\Phi'}$ is a boolean variable that equals one if and only if
  $\last(\rho)\models_{\stratW_\rho}\Phi'$, where $\stratW_\rho$
  is a strategy defined by $\stratW_\rho(\rho')=\stratW(\rho\cdot\rho')$ and $\stratW$ is the window strategy defined above.
  In other words, $y_{\rho,\Phi'}$ encodes "$\Phi'$ is satisfied
  when one follows $\stratW$ after a fixed history of $\rho$".
  \item The probability of satisfaction of PCTL path formulae in
  $\mathcal P$
  is encoded using $\mathcal Z_s=\{z_{\rho,\varphi}\mid
  \rho\in\FPaths_\MDP^{<\Ell}(s), \varphi\in\mathcal P\}$,
  so that $z_{\rho,\varphi}=\mu_{\partial_\rho}(
  \sem{\varphi}_{\partial_\rho}^{\last(\rho)})$, where $\stratW_\rho$
  is a strategy defined by $\stratW_\rho(\rho')=\stratW(\rho\cdot\rho')$.
  In other words, $z_{\rho,\varphi}$ encodes the probability that $\varphi$
  is satisfied
  when one follows $\stratW$ after a fixed history of $\rho$.
\end{itemize}

  Then, $s\models_{\partial} \Phi$ can be encoded schematically
  as the formula
  $\exists \mathcal Y_s,\exists\mathcal Z_s,\text{WellDefined} \mathcal X_s  \wedge \text{WellDefined} \mathcal Y_s  \wedge \text{WellDefined} \mathcal Z_s \wedge \STATER^s(\Phi)$,
  where $\STATER^s(\Phi)$ is defined inductively for every formula in $\mathcal S_0$ by
  \begin{align*}
    \STATER^s(p) &:= \top\text{ if }p\text{ holds on }s, \bot\text{ otherwise}\\
    \STATER^s(\neg p) &:= \bot\text{ if }p\text{ holds on }s, \top\text{ otherwise}\\
    \STATER^s(\Phi_1 \wedge \Phi_2) &:= \STATER^s(\Phi_1) \wedge \STATER^s(\Phi_2)\\
    \STATER^s(\Phi_1 \vee \Phi_2) &:= \STATER^s(\Phi_1) \vee \STATER^s(\Phi_2)\\
    \STATER^s\left(\sum_{i=1}^n c_i \Pr{\varphi_i} \succcurlyeq c_0\right) &:=
      \sum_{i=1}^n c_i z_{s,\varphi_i} \succcurlyeq c_0\,.
  \end{align*}
  In order to ensure that the variables in $\mathcal X_s$ are well defined,
  we let $\text{WellDefined} \mathcal X_s$ be the formula asking
  for every $x_{\rho,a}\in \mathcal X_s$ that $0\leq x_{\rho,a}\leq 1$,
  and for every $\rho\in\FPaths_\MDP^{<\Ell}(s)$ that $\sum_{a\in A} x_{\rho,a}=1$.

  In order to ensure that the variables in $\mathcal Y_s$ are well defined,
  we let $\text{WellDefined} \mathcal Y_s$ be the formula asking
  for every $y_{\rho,\Phi'}\in \mathcal Y_s$ that
  $(y_{\rho,\Phi'} = 0 \vee y_{\rho,\Phi'} = 1) \wedge y_{\rho,\Phi'} = 1 \Leftrightarrow \STATE^\rho(\Phi')$,
  where $\STATE^\rho(\Phi')$ is defined for every formula in $\mathcal S_1$ by
  \begin{align*}
    \STATE^\rho(p) &:= \top\text{ if }p\text{ holds on }\last(\rho), \bot\text{ otherwise}\\
    \STATE^\rho(\neg p) &:= \bot\text{ if }p\text{ holds on }\last(\rho), \top\text{ otherwise}\\
    \STATE^\rho(\Phi_1 \wedge \Phi_2) &:= (y_{\rho,\Phi_1}=1) \wedge (y_{\rho,\Phi_2}=1)\\
    \STATE^\rho(\Phi_1 \vee \Phi_2) &:= (y_{\rho,\Phi_1}=1) \vee (y_{\rho,\Phi_2}=1)\\
    \STATE^\rho\left(\sum_{i=1}^n c_i \Pr{\varphi_i} \succcurlyeq c_0\right) &:=
      \sum_{i=1}^n c_i z_{\rho,\varphi_i} \succcurlyeq c_0\,.
  \end{align*}
  Finally, in order to ensure that the variables in $\mathcal Z_s$ are well defined,
  we let $\text{WellDefined} \mathcal Z_s$ be the formula asking
  for every $z_{\rho,\varphi}\in \mathcal Z_s$ that
  $\PATH^\rho(\varphi)$ holds,
  where $\PATH^\rho(\varphi)$ is defined for every formula in $\mathcal P$ as
  the following expression (where $\ell>0$ and $s=\last(\rho)$)
  \begin{align*}
    \PATH^{\rho}(\XX^{0} \Phi_1) &:= z_{\rho,\XX^0 \Phi_1}=y_{\rho,\Phi_1}\\
    \PATH^{\rho}(\XX^\ell \Phi_1) &:= z_{\rho,\XX^\ell \Phi_1}=\sum_{s\xrightarrow{a}s'} x_{\rho,a} \proba(s,a,s')z_{\rho\cdot as',\XX^{\ell-1} \Phi_1}\\
    \PATH^{\rho}(\Phi_1 \UU^{0} \Phi_2) &:= z_{\rho,\Phi_1 \UU^{0} \Phi_2}=y_{\rho,\Phi_2}\\
    \PATH^{\rho}(\Phi_1 \UU^{\ell} \Phi_2) &:= \begin{cases}
    \text{ if }y_{\rho,\Phi_2}=1\text{ then }z_{\rho,\Phi_1 \UU^{\ell} \Phi_2}=1\\
    \text{ if }y_{\rho,\Phi_2}=0\wedge y_{\rho,\Phi_1}=0\text{ then }z_{\rho,\Phi_1 \UU^{\ell} \Phi_2}=0\\
    \text{ otherwise }z_{\rho,\Phi_1 \UU^{\ell} \Phi_2}=\sum\limits_{s\xrightarrow{a}s'} x_{\rho,a} \proba(s,a,s')z_{\rho\cdot as',\Phi_1 \UU^{\ell-1} \Phi_2}
    \end{cases}\\
    \PATH^{\rho}(\Phi_1 \WW^{0} \Phi_2) &:= \begin{cases}
    \text{ if }y_{\rho,\Phi_1}=1\vee y_{\rho,\Phi_2}=1\text{ then }z_{\rho,\Phi_1 \WW^{0} \Phi_2}=1\\
    \text{ otherwise }z_{\rho,\Phi_1 \WW^{0} \Phi_2}=0
    \end{cases}\\
     \PATH^{\rho}(\Phi_1 \WW^{\ell} \Phi_2) &:= \begin{cases}
     \text{ if }y_{\rho,\Phi_2}=1\text{ then }z_{\rho,\Phi_1 \WW^{\ell} \Phi_2}=1\\
     \text{ if }y_{\rho,\Phi_2}=0\wedge y_{\rho,\Phi_1}=0\text{ then }z_{\rho,\Phi_1 \WW^{\ell} \Phi_2}=0\\
     \text{ otherwise }z_{\rho,\Phi_1 \WW^{\ell} \Phi_2}=\sum\limits_{s\xrightarrow{a}s'} x_{\rho,a} \proba(s,a,s')z_{\rho\cdot as',\Phi_1 \WW^{\ell-1} \Phi_2}
     \end{cases}
  \end{align*}
  These constraints ensure by induction that $z_{\rho,\varphi}$ encodes
  the probability that $\varphi$ is satisfied
  when one follows $\stratW$ after a fixed history of $\rho$.
  Note that conditional formulae such as $\PATH^{\rho}(\Phi_1 \UU^{\ell} \Phi_2)$
  can be defined in the theory of the reals, as a single equality
  of the form $z_{\rho,\Phi_1 \UU^{\ell} \Phi_2}=y_{\rho,\Phi_2}+(1-y_{\rho,\Phi_2})y_{\rho,\Phi_1}\sum \dots$

  The $\exists$-$\R$ formula that we obtain has at most
  $|\MDP|^{\mathcal O(\Ell)}$ free variables,
  and $|\Phi||\MDP|^{\mathcal O(\Ell)}$ variables quantified existentially.
  Overall, the formula has size in
  $|\Phi||\MDP|^{\mathcal O(\Ell)}$.

  To conclude the proof of Proposition~\ref{prop:window-FOR},
  note that the only negation in our formula comes from the
  $y_{\rho,\Phi'} = 1 \Leftrightarrow \STATE^\rho(\Phi')$ subformula
  of $\text{WellDefined} \mathcal Y_s$,
  equivalent in this case to $(y_{\rho,\Phi'} = 0 \vee \STATE^\rho(\Phi'))
  \wedge (y_{\rho,\Phi'} = 1 \vee \neg \STATE^\rho(\Phi'))$.
  If we assume that $\Phi$ is flat, $\mathcal S_1$ only contains formulae
  that use $\wedge$, $\vee$ and atomic propositions, so that
  $\neg \STATE^\rho(\Phi')$ is equivalent to either $\top$ or $\bot$.
  If in addition to being flat, $\Phi$ is also non-strict, then
  every $\succcurlyeq$ in linear inequalities is in fact $\geq$,
  so that our entire $\exists$-$\R$ formula is non-strict.
\end{proof}

\begin{proof}[Proof of Lemma~\ref{lm:fixed-points}]
  We only need to show that $\f$ is Scott-continuous to apply the classical
  fixpoint theorems.
  Let $\mathcal S$ be a directed set of portfolios,
  so that every pair of element in $\mathcal S$ has an upper bound
  and a lower bound in $\mathcal S$.
  Because intersections and unions transfer to the limit
  of a sequence of inclusions, it holds that
  $\bigsqcap \mathcal S$ and $\bigsqcup \mathcal S$ both belong to $\mathcal S$.
  We show that $\f$ is Scott-continuous upwards
  and downwards,\textit{i.e.}~for all $s$, we have
  $\bigsqcup_{\Pi\in \mathcal S} \f(\Pi) =
  \f(\bigsqcup \mathcal S)$ and $\bigsqcap_{\Pi\in \mathcal S} \f(\Pi)=
  \f(\bigsqcap \mathcal S)$.

  Fix a state $s$. A window strategy $\stratW$ is in $\f(\bigsqcap \mathcal S)_s$
  if for all $s\xrightarrow{a}s'$ there exists
  $\stratW'\in \bigcap_{\Pi\in \mathcal S} \Pi_{s'}$ so that $\stratW$
  and $\stratW'$ are compatible w.r.t.~$s\xrightarrow{a}s'$.
  It follows that $\stratW$ is in $\bigcap_{\Pi\in \mathcal S} \f(\Pi)_s$,
  \textit{i.e.}~for all $\Pi\in \mathcal S$,
  for all $s\xrightarrow{a}s'$ there exists
  $\stratW'\in \Pi_{s'}$ so that $\stratW$
  and $\stratW'$ are compatible w.r.t.~$s\xrightarrow{a}s'$.
  Thus $\f(\bigsqcap \mathcal S)\subseteq \bigsqcap_{\Pi\in \mathcal S} \f(\Pi)$.
  Moreover, $\bigsqcap \mathcal S\in \mathcal S$ implies
  $\bigsqcap_{\Pi\in \mathcal S} \f(\Pi)\subseteq \f(\bigsqcap \mathcal S)$.

  Similarly, a window strategy $\stratW$ is in $\bigsqcup_{\Pi\in \mathcal S}
  \f(\Pi)$ if there exists $\Pi\in \mathcal S$ so that,
  for all $s\xrightarrow{a}s'$ there exists
  $\stratW'\in \Pi_{s'}$ so that $\stratW$
  and $\stratW'$ are compatible w.r.t.~$s\xrightarrow{a}s'$.
  It follows that $\stratW$ is in $\f(\bigsqcup \mathcal S)$,
  \textit{i.e.}~for all $s\xrightarrow{a}s'$ there exists
  $\stratW'\in \bigcup_{\Pi\in \mathcal S} \Pi_{s'}$ so that $\stratW$
  and $\stratW'$ are compatible w.r.t.~$s\xrightarrow{a}s'$.
  Moreover, $\bigsqcup \mathcal S\in \mathcal S$ implies
  $\f(\bigsqcup \mathcal S) \subseteq \bigsqcup_{\Pi\in \mathcal S}
  \f(\Pi)$.

  In particular, Scott-continuity implies that $\f$ is monotone:
  $\Pi\subseteq\Pi' \Rightarrow \f(\Pi)\subseteq\f(\Pi')$.
\end{proof}

\begin{proof}[Proof of Proposition~\ref{prop:fixed-point}]
  In order to prove this result we introduce extra notions and lemmas.
  A \emph{window strategy assignment} $\mathcal A$ maps
  every path $\rho$ in $\FPaths_\MDP$ to a window strategy for $s=\last(\rho)$,
  so that %
  for all transitions $s\xrightarrow{a}s'$, $\mathcal A(\rho)$
  and $\mathcal A(\rho\cdot a s')$ are compatible w.r.t.~$s\xrightarrow{a}s'$.
  Intuitively, it can be seen as an extended strategy, that states for a given
  history what the next $\Ell$ decisions will be. The compatibility
  assumption then ensures that these decisions are indeed carried out later on.

  A strategy $\strat$ naturally defines a window strategy assignment
  $\mathcal A^{\strat}$, where $\mathcal A^{\strat}(\rho)$ is
  the window strategy $\stratW_\rho$ defined by $\strat$
  for the fixed prefix $\rho$ and horizon $\Ell$.
  Conversely, a window strategy assignment $\mathcal A$ naturally defines
  a strategy $\strat$ so that for all $\rho$ in $\FPaths_\MDP$,
  the distribution $\strat(\rho)$ equals $\mathcal A(\rho)(\last(\rho))$,
  the first decision of $\mathcal A(\rho)$. Then,
  $\mathcal A=\mathcal A^\strat$.
  Note that a state $s_0$ satisfies $\Forall \GG \Phi$
  under $\sigma$ if and only if for every path $\rho\in\FPaths_\strat(s_0)$
  ending in a state $s$, it holds that $s\models_{\mathcal A(\rho)}\Phi$.
  In particular, there is no requirement for paths that are not reachable
  from $s_0$ when following $\strat$. %

  For a strategy $\strat$ and an initial state $s_0$,
  we define the portfolio $\Pi^{s_0,\strat}$,
  where $\Pi^{s_0,\strat}_s=\{\mathcal A^{\strat}(\rho) \mid
  \rho\in\FPaths_\strat(s_0), \last(\rho)=s\}$.
  In particular, $\Pi^{s_0,\strat}_s$ is empty on any state
  not reachable from $s_0$ under $\strat$, and is
  non-empty on all other states.
  Intuitively, $\Pi^{s_0,\strat}_s$ is the set of window strategies obtained
  for state $s$ and horizon $\Ell$ from fixed prefixes
  of non-zero probability in $\strat$.

\begin{lemma}\label{lm:start-to-fixed-point}
  Let $s_0$ be a state and let $\strat$ be a strategy.
  If $s_0 \models_{\strat} \Forall \GG \Phi$,
  then $\Pi^{s_0,\strat}$ is a fixed point of $\f$ in $Q^\Phi$
  so that $\Pi^{s_0,\strat}_{s_0}\neq\emptyset$.
\end{lemma}
\begin{proof}[Proof of Lemma~\ref{lm:start-to-fixed-point}]
  First, note that $\Pi^{s_0,\strat}_{s_0}$ cannot be empty,
  as it must contain at least $\mathcal A^\sigma(s_0)$.
  We show that $\Pi^{s_0,\strat}\in Q^\Phi$
  and $\f(\Pi^{s_0,\strat})=\Pi^{s_0,\strat}$.
  By definition of $\Forall \GG \Phi$ and Lemma~\ref{lm:window-length},
  every window strategy
  $\stratW_\rho=\mathcal A^{\strat}(\rho)$
  (defined from $\strat$ by a prefix $\rho$ of non-zero probability in $\strat$)
  satisfies $\Phi$, so that $\Pi^{s_0,\strat}\in Q^\Phi$.
  For every transition $s\xrightarrow{a}s'$
  so that $\stratW_\rho(s,a)>0$,
  we have that $\stratW_\rho$ and $\mathcal A^{\strat}(\rho\cdot as')\in\Pi^{s_0,\strat}_{s'}$
  are compatible w.r.t.~$s\xrightarrow{a}s'$,
  and hence $\stratW_\rho\in \f(\Pi^{s_0,\strat})$.
  Thus $\Pi^{s_0,\strat}\subseteq\f(\Pi^{s_0,\strat})$,
  and the other inclusion holds by definition of $\f$.
\end{proof}

\begin{lemma}\label{lm:fixed-point-to-strat}
  From every fixed point $\Pi$ of $\f$ in $Q^\Phi$ that is
  non-empty on a state $s_0$, we can extract a
  window strategy assignment $\mathcal A$ that defines a strategy $\sigma$
  so that $s_0 \models_{\strat} \Forall \GG \Phi$.
\end{lemma}
\begin{proof}[Proof of Lemma~\ref{lm:fixed-point-to-strat}]
  Assume $\Pi_{s_0}\neq \emptyset$. We construct
  for each path $\rho$ in $\MDP$ %
  a window strategy $\stratW_\rho$,
  and set $\mathcal A(\rho)=\stratW_\rho$.
  We proceed by induction on the length
  of $\rho$, and only define $\stratW_\rho$ for paths $\rho$
  starting from $s_0$ with non-zero probability in $\sigma$,
  the strategy defined by $\mathcal A$.
  The assignment $\mathcal A$ can be arbitrarily defined on every other path, %
  \textit{e.g.}~with an uniform strategy.
  For all paths $\rho$ from $s_0$ to a state $s$ of non-zero
  probability under $\sigma$,
  we ensure that $\stratW_\rho\in\Pi_{s}$.
  Then, $\stratW_{\rho}\in \Pi_{s}$ and $\Pi\in Q^\Phi$ implies
  $s \models_{\stratW_{\rho}} \Phi$.
  Overall, we obtain that $s_0$ satisfies $\Forall \GG \Phi$
  when following the assignment $\mathcal A$.
  Let us now describe the induction.
  For the first move, we pick $\stratW_{s_0}$ as
  any window strategy in $\Pi^\Phi_{s_0}\neq\emptyset$.
  Assume that $\stratW_\rho$ is defined for a path $\rho$ that ends in $s$,
  and that $\proba(s,a,s')>0$, and that $\stratW_\rho(s,a)>0$. We define
  $\stratW_{\rho \cdot a s'}$ as follows.
  As $\stratW_\rho\in \Pi_s = \f(\Pi)_s$ (fixed point),
  we know by definition of $\f$ that
  there exists $\stratW'\in\Pi_{s'}$ such that $\stratW_\rho$ and
  $\stratW'$ are compatible w.r.t.~$s\xrightarrow{a}s'$.
  We then let $\stratW_{\rho \cdot a s'} = \stratW'$.
\end{proof}

As a corollary of Lemmas~\ref{lm:start-to-fixed-point}
and \ref{lm:fixed-point-to-strat}, we obtain Proposition~\ref{prop:fixed-point}.
\end{proof}

\begin{proof}[Proof of Theorem~\ref{thm:coRE-easy}]
  Let $n=|\mathcal X_{s_0}|$.
  We start by showing that the limit of an infinite decreasing sequence of non-empty
  compact sets in $\R^n$ is non-empty.
  Consider a sequence $(\mathcal S_i)_{i\in\N}$ of non-empty
  compact sets in $\R^{n}$, so that for all $i\in\N$,
  $\mathcal S_{i+1}\subseteq \mathcal S_{i}$.
  Let us prove $\bigcap_{i\in\N} \mathcal S_i\neq\emptyset$.
  For each $i\in\N$, we can pick some point $v_i\in\R^{n}$ in $\mathcal S_i$,
  as it is not empty.
  Consider now the sequence $(v_i)_{i\in\N}$. Since $(\mathcal S_i)_{i\in\N}$
  is decreasing, every $v_i$ belongs to $\mathcal S_0$.
  Since $\mathcal S_0$ is a compact set, there exists a convergent subsequence
  $(v_{n_j})_{j\in\N}$ which converges to a point $v$ in $\mathcal S_0$.
  Let us show that $v\in\mathcal S_i$ for every $i\in\N$, so that
  $v\in\bigcap_{i\in\N} \mathcal S_i$. Fix $i\in\N$.
  For all $j\geq i$, it holds that $n_j\geq j\geq i$. Therefore, the subsequence $(v_{n_j})_{j\geq i}$ is in the compact set $\mathcal S_i$,
  and its limit $v$ must belong to $\mathcal S_i$.
  Thus, $\bigcap_{i\in\N} \mathcal S_i\neq\emptyset$.

  Therefore, if there is no strategy $\sigma$ so that
  $s \models_{\strat} \Forall \GG \Phi$ then
  there exists $i\in\N$ so that $\f^i(\bigsqcup Q^\Phi)_{s_0}=\emptyset$.
  As $\f^i(\bigsqcup Q^\Phi)_{s_0}=\emptyset$
  is expressible in the theory of the reals for all $i \in \N$,
  and is therefore decidable,
  we can enumerate $i$ and obtain a semi-algorithm for the problem asking
  if there is no strategy $\strat$ so that
  $s_0 \models_{\strat} \Forall \GG \Phi$. Then, the complement
  of the synthesis problem for flat, non-strict global window $L$-PCTL
  formulae is in $\RE$.
\end{proof}

\begin{proof}[Proof of Theorem~\ref{thm:satisfiability}]
  Given a bound $N$ and a global window formula $\Forall \GG \Phi$
  over atomic propositions $\AP$,
  we construct an MDP
  $\MDP=\zug{S,A,\initState, \proba, \AP\uplus\{\initState\}, L}$ and a window formula $\Phi'$
  so that there is a deterministic strategy in $\MDP$ that satisfies $\Forall \GG \Phi'$
  if and only if there exists an MC of granularity bounded by $N$ that satisfies $\Forall \GG \Phi$.

  We note that if an MC $M$ has granularity bounded by $N$ then
  all probabilities in $M$ are multiples of $\frac{1}{N!}$,
  so that the outdegree of $M$ is bounded by $N!$.
  Let $\Dist_N$ denote the finite set of distributions, on the set $2^\AP\times\{1,\dots,N!\}$, that satisfy the following constraints.
  If $d\in\Dist_N$, then every probability in $d$ must be of granularity
  bounded by $N$. Moreover, there exists $k\leq N!$ and $s_1\dots s_k\in 2^\AP$
  so that the support of $d$ equals $\{(s_1,1),\dots,(s_k,k)\}$, so that the support of $d$
  describes an ordered sequence of elements of $2^\AP$.
  Our MDP $\MDP$ has states $S=\{\initState\}\uplus (2^\AP\times\{1,\dots,N!\})$,
  initial state $\initState$, actions $A=\Dist_N\uplus (2^\AP\times\{1,\dots,N!\})$ and a labelling function
  that maps $\initState$ to $\initState$ and every state $(s,i)$ to $s$.
  We assume w.l.o.g. that the actions in $(2^\AP\times\{1,\dots,N!\})$ are available for $\initState$ only,
  while the actions in $\Dist_N$ are available for every other state.
  The transition relation is defined from the actions, so that for every state $(s,i)$ in $S$
  and every action $d$ in $\Dist_N$, $\proba((s,i),d)=d$.
  In the initial state the action picks a state $(s,i)$ to go to with probability $1$.

  On the one hand, every deterministic strategy in $\MDP$ defines an MC
  of granularity bounded by $N$.
  On the other hand, every MC $M$ of granularity bounded by $N$
  can be unfolded into a tree of states $(l,i)\in 2^\AP\times\{1,\dots,N!\}$,
  so that $l$ is the label of the corresponding state in $M$ and $i$
  is a "siblings index", so that the children $s_1\dots s_k$ of a node
  in $M$ become $(s_1,1),\dots,(s_k,k)$ in the unfolding.
  Then, for every MC $M$ of granularity bounded by $N$
  there exists a deterministic strategy in $\MDP$ that mimics it, in the sense that
  the first decision goes from $\initState$ to $(l,1)$
  with $l$ the label of the initial state of $M$, and the following decisions
  construct the unfolding of $M$ mentioned above by picking the appropriate distributions
  in $\Dist_N$.

  If we consider the formula $\Phi'$ defined as $\initState \vee \Phi$,
  it follows that there exists a deterministic $\sigma$ so that
  $\initState\models_\sigma \Forall\GG \Phi'$ in $\MDP$ if and only if
  there exists an MC $M$ of granularity bounded by $N$
  so that $M\models \Forall\GG \Phi$.
  We note that $|\MDP|$ is exponential in $|\Phi|$ and $N$, while $|\Phi'|=O(|\Phi|)$.
  Applying the procedure of Proposition~\ref{prop:HD-global-window} will then have complexity in $2^{|\MDP|^{\mathcal O(|\Phi'|)}}$, and therefore doubly-exponential in $|\Phi|$ and $N$.

  When we translate between strategies and MCs
  in this reduction, we can make the observation that finite MCs
  correspond to a particular kind of deterministic strategies: so called finite-memory strategies,
  whose behaviour can be described by a deterministic finite automaton.
  Moreover, the decision procedure from Proposition~\ref{prop:HD-global-window}
  always yield finite memory strategies, as it suffices to remember
  the fixed point $\Pi^\Phi$ and the last $\Ell$ transitions in order to determine
  what the next decision should be.
  Overall, our approach always produces finite MCs.
\end{proof}

\section{Probabilistic noninterference.} \label{app:noninterference}
Hyperproperties~\cite{CS10} consist of specifications that describe the correctness of a system with respect to multiple executions.
Probabilistic noninterference and observational determinism \cite{GM82,ZM03} are examples of hyperproperties that concerns information leaks through cohort channels.
We consider a specific such example from \cite{AB18} which states that a low-privileged user (e.g. an attacker) should not be able to distinguish two computations which differ in their inputs by high-privileged users.
The probabilistic noninterference property requires that given a Markov chain and two equivalent initial states as observed by a low-privileged user in terms of their information content, an equivalence in terms of information level is maintained throughout the computations starting from such initial states.

To formally express probabilistic noninterference, we consider an atomic proposition $l$ denoting low-information observability. Thus for some states, the proposition $l$ holds true, while for the rest, the proposition $l$ does not hold true.
Given a Markov chain $M$ with a set $S$ of states, a set ${\sf AP}=\{l\}$ of the single atomic proposition $l$, a labelling function $L: S \rightarrow 2^{\sf AP}$, to express probabilistic noninterference, we construct a new Markov chain $\widehat{M}$ which consists of the product of $M$ with itself, and a new initial state $\widehat{s}_{\sf init}$ with transitions to all other states of $\widehat{M}$ with equal probabilities.
The other states of $\widehat{M}$ are of the form $(s_i, s_j')$ for $s_i, s_j \in S$.
Let ${\sf AP}' = \{l'\}$, and $L' : S \rightarrow 2^{{\sf AP}'}$
is defined as $L'(s) = \{l'\}$ if $L(s) = l$, else $L'(s)=\emptyset$.
Now for a state $(s_i, s'_j)$ of $\widehat{M}$, we define the labelling function $\widehat{L}(s_i,s'_j) = L(s_i) \cup L'(s_j)$, and for the initial state $\widehat{s}_{\sf init}$, we have $\widehat{L}(\widehat{s}_{\sf init})=\emptyset$.
Thus probabilistic noninterference can be defined by the formula $\Phi_{\sf pni}$ which is as follows.
\[
\Phi_{\sf pni}= \PrBig{\XX \: [l \wedge l' \implies [\Pr{(\XX \: l)}=\Pr{(\XX \: l')] \:} }=1
\]
In the context of an MDP, one can ask the question if for all possible attack policies of an attacker, $\Phi_{\sf pni}$ is \emph{not} violated.
Stated otherwise, given an MDP $\MDP$, we ask if there exists a strategy $\sigma$ of the attacker such that for traces starting from some pair of initial states that are low-equivalent, the probabilistic noninterference property is violated, i.e. $\neg \Phi_{\sf pni}$ holds true.

\section{Complements of Section~\ref{sec:undec}}
\label{app:undecidability}

To each instruction, corresponds one gadget. We also need one additional gadget to initialize the values of the counters to 0. Overall, the reduction takes polynomial time in the number of instructions of the Minsky machine. Furthermore, the maximum window length used in all the reduction is 4, regardless of the Minsky machine taken as input.

In the following, when depicting gadgets, we will use the notations $\sigma_1(\cdot)$ or $\sigma_2(\cdot)$ to denote a choice of the strategy $\sigma$ to take a given action. Note that this is informal as a strategy $\strat$ has type $\strat: \FPaths_{\MDP} \to \Dist(A)$. In all occurrences of $\sigma_1(\cdot)$ or $\sigma_2(\cdot)$, the choice of the strategy depends on the complete history, not only on the current gadget it is in.

In the remainder of this section, we describe each gadget (besides the ones already presented in the main part of the paper) used in the reduction.

\subsection{Proof for flat non-strict formula}
\label{app:undecidability_flat_non_strict}
\subsubsection{Initial gadget}
Let us consider the initial gadget $\mathsf{Gd}_\mathsf{init}$ which is depicted in Figure~\ref{InitialGadget}. This gadget  initializes the values of the counters to 0. That is, its exiting probability is $\frac{5}{6}$. This is stated in the proposition below.

\begin{figure}
	\centering
	\includegraphics[scale=0.9]{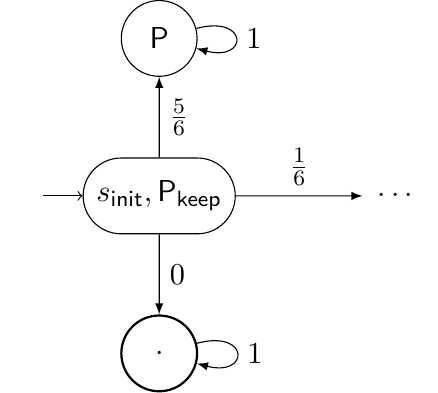}
	\caption{The initial gadget $\mathsf{Gd}_\mathsf{init}$.}
	\label{InitialGadget}
\end{figure}

\begin{proposition}[Exiting probability of the initial gadget]
	The exiting probability $p_{\mathsf{Gd}_{\mathsf{init}}}^{\mathsf{ex}}$ of the initial gadget $\mathsf{Gd}_\mathsf{init}$ is equal to $p_{\mathsf{Gd}_{\mathsf{init}}}^{\mathsf{ex}} = \frac{5}{6} = p(0,0)$.
	\label{prop:initial_gadget}
\end{proposition}

Note that the gadget visited after the initial gadget is the one corresponding to the first instruction. The sequence of gadgets visited then exactly corresponds to the sequence of instructions seen in the Minsky machine.

\subsubsection{Gadget testing that the entering probability is 5/6}
Let us now consider the gadget $\mathsf{Gd}_{p = \frac{5}{6}}$ testing that the values of the counters is 0, i.e. that the entering probability is $\frac{5}{6}$. Note that this gadget is used to test that the value of counters is 0 or that it is positive. The gadget $\mathsf{Gd}_{p = \frac{5}{6}}$ is depicted in Figure~\ref{GadgetFiveSix}. We additionally define the formula $\Phi_{= 5/6} = \mathsf{P}_{5/6} \Rightarrow \mathbb{P}(\FF^{2} \mathsf{P}_{5/6}') = \frac{5}{6}$. These definitions ensure the following proposition:
\begin{proposition}[Gadget testing equality to $5/6$ Specification]
	The formula $\Forall \GG \Phi_{= 5/6}$ is satisfied in gadget $\mathsf{Gd}_{= \frac{5}{6}}$ if and only if the entering probability
	$p_{\mathsf{Gd}_{=\frac{5}{6}}}^{\mathsf{en}}$ of the gadget $\mathsf{Gd}_{= \frac{5}{6}}$ is equal to $\frac{5}{6}$.
	\label{prop:gadgetFiveSix}
\end{proposition}

\begin{figure}
	\centering
	\includegraphics[scale=0.9]{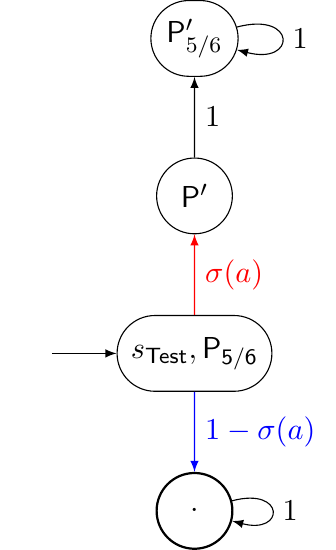}
	\caption{The gadget $\mathsf{Gd}_\mathsf{=5/6}$ testing the equality to $5/6$.}
	\label{GadgetFiveSix}
\end{figure}

\subsubsection{Gadget testing that a counter value is positive}
Consider now the gadget $\mathsf{Gd}_{c_2 > 0}$ testing that the value of the counter $c_2$ is positive. It is depicted in Figure~\ref{fig:GadgetTestNeq}.

\begin{figure}
	\centering
	\includegraphics[scale=1]{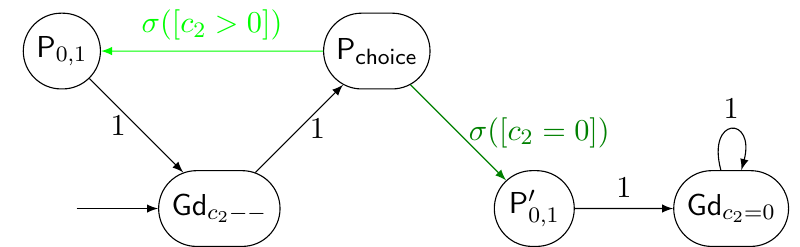}
	\caption{The gadget $\mathsf{Gd}_{c_2 > 0}$ for testing that the counter $c_2 > 0$.}
	\label{fig:GadgetTestNeq}
\end{figure}

The idea is the following: when the gadget is entered, the counter $c_2$ is decremented once. Then, when the predicate $\mathsf{P}_{\mathsf{choice}}$ is seen, the strategy $\sigma$ is given a deterministic (this is ensured by formula $\Phi_{0,1}$) choice: either it takes the edge on the left, and the counter $c_2$ is decremented again, or it takes the edge on the right to test that the value of the counter $c_2$ now is 0 (with the gadget $\mathsf{Gd}_{c_2 = 0}$). Hence, if the value of counter $c_2$ when entering the gadget is positive, it can be decremented to reach $0$, thus ensuring that no issue will arise when visiting the gadget $\mathsf{Gd}_{c_2 = 0}$. However, if this value is 0, since it is decremented once while entering the gadget, there is no way to visit the gadget $\mathsf{Gd}_{c_2 = 0}$ without raising an issue. The specification of this gadget is given below.

\begin{proposition}[Testing Positive Value Gadget Specification]
	Assume that the entering probability $p_{\mathsf{Gd}_{c_2 > 0}}^{\mathsf{en}}$ of the gadget $\mathsf{Gd}_{c_2 = 0}$ is equal to $p(x_2,x_3)$ for some $x_2,x_3 \in \mathbb{N}$. Then, there is a strategy $\sigma$ such that the formula $\Forall \GG [\Phi_{\mathsf{keep}} \wedge \Phi_{0,1} \wedge \Phi_{c_2--} \wedge \Phi_{c_3--} \wedge \Phi_{= 5/6}]$ is satisfied if and only if $x_2 > 0$.
	\label{prop:testing_gadget_neq}
\end{proposition}

\subsubsection{Duplicating gadget}
Before considering the branching gadget, we need to be able to duplicate probabilities. Indeed, in the branching gadget, one probability is used to test that indeed the choice of the strategy is coherent with the values of the counters, the other is used to continue in the MDP. Hence, we need a gadget with two exiting probabilities that are both equal to the exiting probability of the previous gadget. This is done with the duplicating gadget $\mathsf{Gd}_\mathsf{dupl}$ depicted in Figure~\ref{fig:GadgetDupl}.

\begin{figure}
	\centering
	\includegraphics[scale=1]{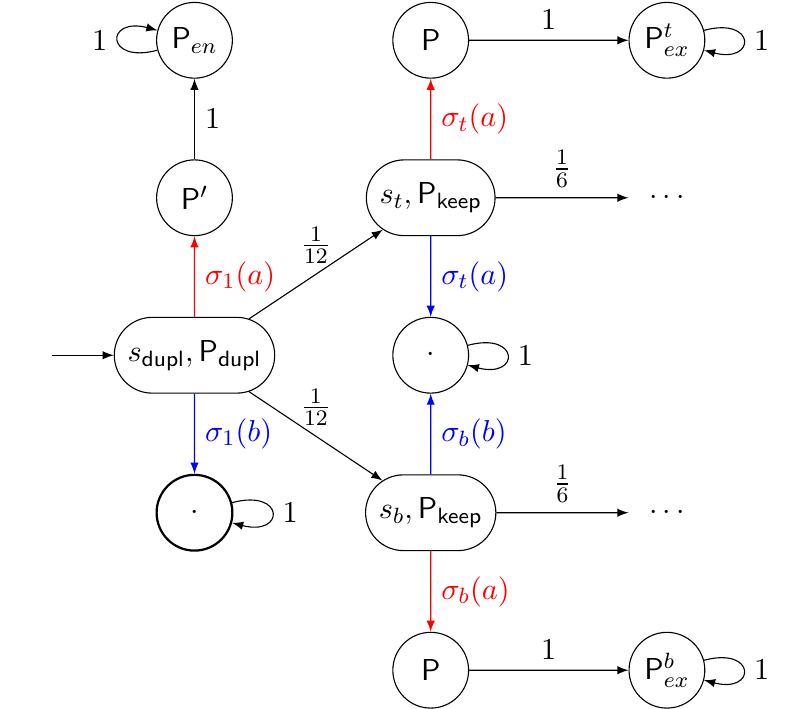}
	\caption{The gadget $\mathsf{Gd}_{\mathsf{dupl}}$ for duplicating probabilities.}
	\label{fig:GadgetDupl}
\end{figure}

We additionally need the following formula.
\begin{displaymath}
	\Phi_{\mathsf{dupl}} := \mathsf{P}_{\mathsf{dupl}} \Rightarrow (\mathbb{P}(\FF^{2} \mathsf{P}_{en}) = 12 \cdot \mathbb{P}(\FF^{3} \mathsf{P}^t_{ex}) \wedge \mathbb{P}(\FF^{2} \mathsf{P}_{en}) = 12 \cdot \mathbb{P}(\FF^{3} \mathsf{P}_{ex}^b)
\end{displaymath}

These definitions ensure the proposition below.
\begin{proposition}[Duplicating gadget Specification]
    The duplicating gadget $\mathsf{Gd}_\mathsf{dupl}$ has two exiting probabilities $p_{\mathsf{Gd}_\mathsf{dupl}}^{\mathsf{ex},t}$ (from state $s_t$) and $p_{\mathsf{Gd}_\mathsf{dupl}}^{\mathsf{ex},b}$ (from state $s_b$). Assume that this gadget $\mathsf{Gd}_\mathsf{dupl}$ is well-placed and follows a gadget $\mathsf{Gd}$. Then, for a strategy $\sigma$ such that the formula $\Forall \GG [\Phi_\mathsf{keep} \wedge \Phi_{\mathsf{dupl}}]$ is satisfied, then, the exiting probability of gadget $\mathsf{Gd}$ is equal to the exiting probabilities of gadget $\mathsf{Gd}_\mathsf{dupl}$:
	\begin{displaymath}
		p_{\mathsf{Gd}_\mathsf{dupl}}^{\mathsf{ex},t} = p_{\mathsf{Gd}_\mathsf{dupl}}^{\mathsf{ex},b} = p_{\mathsf{Gd}}^{\mathsf{ex}}
	\end{displaymath}
	\label{prop:gadget_dupl}
\end{proposition}

\subsubsection{Branching Gadget}
We can now consider the branching gadget. Specifically, consider an instruction $\mathsf{Branch}_2(k,m)$. Then the corresponding gadget $\mathsf{Gd}_{\mathsf{B}(2,k,m)}$ is depicted in Figure~\ref{fig:GadgetBranching}. The idea is that, once a deterministic choice is made and we are either at state $s_m$ or $s_k$, we duplicate the probabilities to check that the choice was coherent with the values of the counters. Specifically, if $s_m$ is reached, then one must check that the value of counter $c_2$ is positive and if $s_k$ is reached, then one must check that the value of counter $c_2$ is 0.

\begin{figure}
	\centering
	\includegraphics[scale=1]{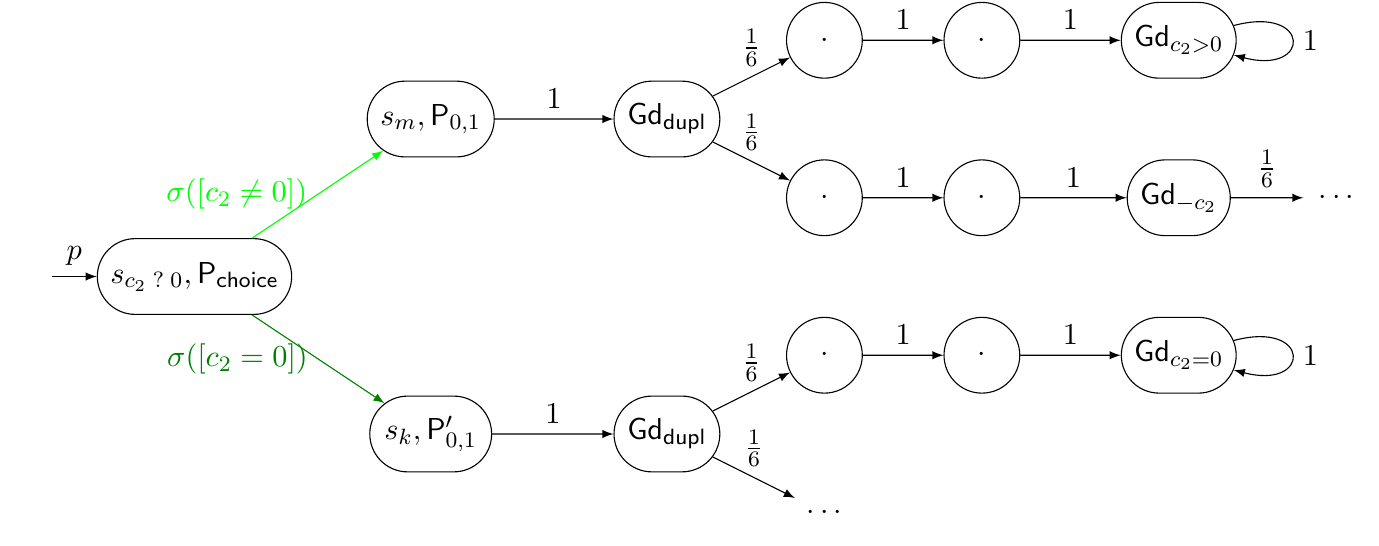}
	\caption{The gadget $\mathsf{Gd}_{\mathsf{dupl}}$ for duplicating probabilities.}
	\label{fig:GadgetBranching}
\end{figure}

This gadget ensures the following specification.
\begin{proposition}[Gadget $\mathsf{Gd}_{\mathsf{B}(2,k,m)}$ Specification]
		 Assume that the gadget $\mathsf{Gd}_{\mathsf{B}(2,k,m)}$ is well-placed and follows a gadget $\mathsf{Gd}$. If the exiting probability $p^{ex}_{\mathsf{Gd}}$ is equal to $p(x_2,x_3)$ for some $x_2,x_3 \in \N$, then for a strategy $\strat$ such that the formula $\Forall \GG [\Phi_\mathsf{keep} \wedge \Phi_{0,1} \wedge \Phi_{\mathsf{dupl}} \wedge \Phi_{c_2--} \wedge \Phi_{c_3--} \wedge \Phi_{= 5/6}]$ is satisfied, then:
		\begin{itemize}
			\item if $x_2 = 0$, then the state $s_k$ is reached and $p^{ex}_{\mathsf{Gd}_{\mathsf{B}(2,k,m)}} = p(x_2,x_3)$;
			\item if $x_2 > 0$, then the state $s_m$ is reached and $p^{ex}_{\mathsf{Gd}_{\mathsf{B}(2,k,m)}} = p(x_2-1,x_3)$.
		\end{itemize}
		\label{prop:mdp_test}
	\end{proposition}

\subsubsection{Halting gadget}
Finally, the instruction $\mathsf{Halt}$ is translated into a gadget in which the formula $\Forall \GG \Phi_{\mathsf{Keep}}$ cannot be satisfied, regardless of the strategy $\sigma$ considered. Such a gadget is depicted in Figure~\ref{GadgetHalt}.
\begin{figure}
	\centering
	\includegraphics{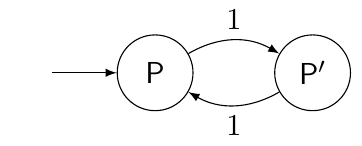}
	\caption{The gadget $\mathsf{Gd}_\mathsf{Halt}$.}
	\label{GadgetHalt}
\end{figure}

\begin{proposition}[Halting Gadget Specification]
	There is no strategy $\strat$ such that the formula $\Forall \GG \Phi_\mathsf{keep}$ is satisfied in the gadget $\mathsf{Gd}_\mathsf{Halt}$.
	\label{prop:halt}
\end{proposition}

\subsection{Proof for flat formulae}
\label{app:undecidability_flat}
If we consider arbitrary flat formulae, then the synthesis problem for global window $L$-PCTL is
$\Sigma_1^1$-hard. Indeed, we reduce the problem of existence of an (infinite) execution in a non-deterministic two-counter Minsky machine that visits infinitely often the first instruction, denoted $\mathsf{MinskyVisit}_\infty$, which is highly undecidable, as stated in the theorem below.

\begin{theorem}[\cite{AH94}]
	\label{thm:minsky_sigma_one_one_hard}
	$\mathsf{MinskyVisit}_\infty$ is $\Sigma_1^1$-hard.
\end{theorem}

In a non-deterministic Minsky machine, the halt instruction is replaced by a non-deterministic choice: $\mathsf{Choice}(k,m): goto \{ k,m \}$. We adapt our reduction to take into account this new instruction and semantics. %

Consider a non-deterministic Minsky machine $\mathsf{Msk} = l_1:\mathsf{ins}_1,\ldots,l_n:\mathsf{ins}_n$ on two counters $c_2$ and $c_3$. First, we need to be able to encode the new instruction $\mathsf{Choice}(k,m)$. This is handled straightforwardly as it only consists in giving the strategy a choice between two branches. The choice needs to be deterministic (i.e. with probability 0,1). This kind of issue is already handled in the gadget $\mathsf{Gd}_{c_2 = 0}$ from Figure~\ref{GadgetTestEq0}. In addition, we have to handle the new specification of the Minsky machine: a positive instance is not one that does not halt but one that visits infinitely often the first instruction. The idea is that, whenever the first instruction is visited, the strategy is given a choice that will correspond to the (finite) number of steps necessary to see the first instruction again. The number of steps $x$ is encoded in a new counter $c_5$ with a probability $p_5 = \frac{1}{5^x}$. Since the number of steps needed to see the first instruction again may be arbitrarily high, the probability $p_5$ may be arbitrarily close to 0, hence, we need to consider a non-strict formula to specify that $p_5 > 0$. Then, each time an instruction that is not the first one is seen, this counter is decremented and an issue arises if such a decrement is not possible.

The addition of this new counter $c_5$ changes a few things. For instance, when testing the equality of a counter to 0, we have to also give the possibility to decrement the counter down to 0 (just like it is done for counter $c_3$ in gadget $\mathsf{Gd}_{c_2++}$). However, note that these changes are minor and do not change the idea of the construction.

We give here the gadget to initialize the value of the counter to an arbitrary high (but still finite) value. The gadget $\mathsf{Gd}_{c_5}$ initializing this counter is depicted in Figure~\ref{fig:NewGadget}.
\begin{figure}
	\centering
	\includegraphics{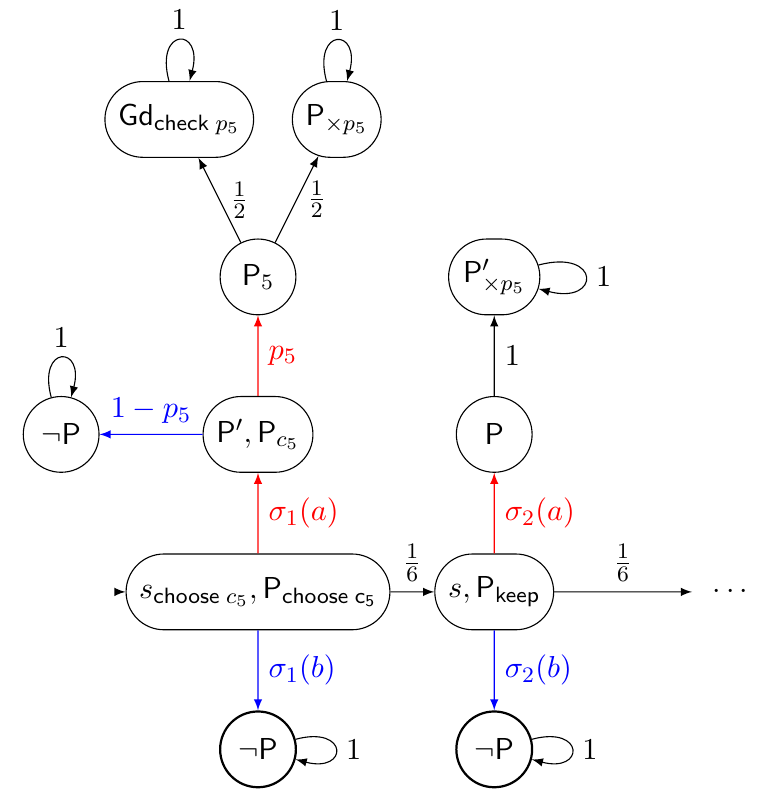}
	\caption{The gadget $\mathsf{Gd}_{c_5}$ for initializing the counter $c_5$.}
	\label{fig:NewGadget}
\end{figure}

The gadget $\mathsf{Gd}_{\mathsf{check \; c_5}}$ is here to check that $p_5 = \frac{1}{5^x}$ for some $x > 0$. Note this is done similarly to the gadget $\mathsf{Gd}_{c_2 = 0}$ from Figure~\ref{GadgetTestEq0} by considering the probability to go from a state ensuring $\mathsf{P}_{c_5}$ to a state ensuring $\mathsf{P}_5$. Then, we consider the formula $\Phi_{c_5}$ below which ensures that the new running probability to see $\mathsf{P}$ has been multiplied by $p_5$:
\begin{displaymath}
\Phi_{c_5} := \mathsf{P}_{\mathsf{choose} \; c_5} \Rightarrow (\mathbb{P}(\FF^3 \mathsf{P}_{\times \; p_5}) > 0 \wedge \mathbb{P}(\FF^3 \mathsf{P}_{\times \; p_5}) = 3 \cdot \mathbb{P}(\FF^3 \mathsf{P}_{\times \; p_5}'))
\end{displaymath}

Note that this formula uses a strict comparison of probabilities. These definitions ensure proposition below.
\begin{proposition}[Gadget Initializing $c_5$ Specification]
	The values that the exiting probability $p_{\mathsf{Gd}_{c_5}}^{\mathsf{ex}}$ can take under the condition that the formula $\Forall \GG \Phi$ is satisfied (where $\Phi$ is a conjunction of the formulae defined so far including $\Phi_{c_5}$) are exactly $\{ p_{\mathsf{Gd}_{c_5}}^{\mathsf{en}}\cdot \frac{1}{5^n} \mid n \in \mathbb{N} \}$.
	\label{prop:init_c5_gadget}
\end{proposition}

\end{document}